\documentclass[11pt,cls,onecolumn]{IEEEtran}

\usepackage{stfloats}
\usepackage{color}
\usepackage{amssymb} 
\usepackage{mathtools}
\usepackage{footnote}
\usepackage{amsfonts}
\usepackage[colorlinks=true, linkcolor=blue]{hyperref}
\usepackage{booktabs}
\usepackage{makecell}

\usepackage{stmaryrd}  

\newcommand{\ve}[1]{{\bf #1}}

\newcommand{\mc}{\mathcal}
\newcommand{\wcf}{\mathcal F_m^{\text{W}}}
\newcommand{\bfm}{\mathcal F_m}
\newcommand{\intset}[1]{[\![#1]\!]}
\newcommand{\lt}{\mathsf{T}}
\newcommand{\lf}{\mathsf{F}}
\newcommand{\cbfc}{C_{\text{BFC}}}

\newtheorem{lemma}{Lemma}
\newtheorem{theorem}{Theorem}
\newtheorem{proposition}{Proposition}
\newtheorem{definition}{Definition}

\newtheorem{remark}{Remark}

\newtheorem{example}{Example}

\ifCLASSOPTIONcompsoc \usepackage[caption=false,font=normalsize,labelfon t=sf,textfont=sf]{subfig} \else \usepackage[caption=false,font=footnotesize]{subfi
g} \fi

\begin{document}

\sloppy
\IEEEoverridecommandlockouts
%% Paper Title
%% You can use linebreaks \\ within to get better formatting as
%% desired. 
\title{Beyond Identification: Computing Boolean Functions via Channels}

%% Author names and affiliations:
%%
%% Avoiding spaces at the end of the author lines is not a problem with
%% conference papers because we don't use \thanks or \IEEEmembership.
%%
%% For several authors with only one affiliation:
%%
\author{ Jingge Zhu, Matthias Frey\\

%Department of Electronic and Electrical Engineering\\

%The University of Melbourne
%
%\thanks{This work was supported in part by the European ERC Starting Grant 259530-ComCom. }
\thanks{J. Zhu and M. Frey are  with Department of Electronic and Electrical Engineering, The University of Melbourne,  Parkville, Victoria, Australia. (e-mail: \{jingge.zhu, matthias.frey\}@unimelb.edu.au)} 
%Switzerland (e-mail: jingge.zhu@epfl.ch).}
%
%
%\thanks{M. Gastpar is with the School of Computer and Communication Sciences, Ecole Polytechnique F{\'e}d{\'e}rale de Lausanne (EPFL), Lausanne, Switzerland and the Department of Electrical Engineering and Computer Sciences, University of California, Berkeley, CA, USA (e-mail: michael.gastpar@epfl.ch).}
}

%%
%% For over three affiliations, or if they all won't fit within the width
%% of the page, use this alternative format:
%%
% \author{
%   \IEEEauthorblockN{
%     Michael Shell\IEEEauthorrefmark{1},
%     Homer Simpson\IEEEauthorrefmark{2},
%     James Kirk\IEEEauthorrefmark{3}, 
%     Montgomery Scott\IEEEauthorrefmark{3} and
%     Eldon Tyrell\IEEEauthorrefmark{4}}
%   \IEEEauthorblockA{
%     \IEEEauthorrefmark{1}School of Electrical and Computer Engineering\\
%     Georgia Institute of Technology, Atlanta, Georgia 30332--0250\\ 
%     Email: see http://www.michaelshell.org/contact.html}
%   \IEEEauthorblockA{
%     \IEEEauthorrefmark{2}Twentieth Century Fox, Springfield, USA\\
%     Email: homer@thesimpsons.com}
%   \IEEEauthorblockA{
%     \IEEEauthorrefmark{3}Starfleet Academy, San Francisco, California 96678-2391\\
%     Telephone: (800) 555--1212, Fax: (888) 555--1212}
%   \IEEEauthorblockA{
%     \IEEEauthorrefmark{4}Tyrell Inc., 123 Replicant Street, Los Angeles, California 90210--4321}
% }

%% Use for special paper notices

%% To balance the two columns, you should reduce the text-height of
%% the last page using the following command:
%%%%%%%%%%%%%%%%%%%%%%%%%%%%%%%%%%%%%%%%%%%%%%%%%%%%%%%%%%%%%%%%%%%%%
%\addtolength{\textheight}{-9.35cm}
%%%%%%%%%%%%%%%%%%%%%%%%%%%%%%%%%%%%%%%%%%%%%%%%%%%%%%%%%%%%%%%%%%%%%
%% with an appropriate value. This command must be place on the second
%% last page, i.e., for a one-page abstract here, for a two-page
%% abstract right after the \maketitle command.

%% Create the title:
\maketitle

%% Abstract: 
%% For the final version of the accepted paper, please make sure you
%% remove the comment "THIS PAPER IS ELIGIBLE FOR THE STUDENT PAPER
%% AWARD."
%%
\begin{abstract}
Consider a point-to-point communication system in which the transmitter holds a binary message of length $m$ and transmits a corresponding codeword of length $n$.  The receiver's goal is to recover a Boolean function of that message, where the function is unknown to the transmitter, but chosen from a known class $\mc F$. We are interested in the asymptotic relationship of $m$ and $n$:  given $n$, how large can $m$ be (asymptotically), such that the value of the Boolean function can be recovered reliably? This problem generalizes the identification-via-channels framework introduced by Ahlswede and Dueck. We formulate the notion of computation capacity, and derive achievability and converse results for selected classes of functions $\mc F$, characterized by the Hamming weight of functions. Our obtained results are tight in the sense of the scaling behavior for all cases of $\mc F$ considered in the paper.  
\end{abstract}

\section{Introduction}

We consider the problem of computing Boolean functions over noisy channels. To motivate the problem, consider the following simple example of a Battery Management System in an electric vehicle,  where a sensing unit reports a binary status vector  $\ve b=(\ve b_1,\ve b_2,\ve b_3)\in\{0,1\}^3$,
with $\ve b_1$, $\ve b_2$, and $\ve b_3$ indicating over--voltage, under--voltage, and
over--temperature conditions of the battery, respectively.  Depending on the operational objective, the controller does not require the full vector $\ve b$, but instead only a Boolean function of $\ve b$, transmitted through a possibly noisy channel.  For example, in \textit{driving mode} the controller needs to evaluate 
\begin{align*}
    f_{\mathrm{drive}}(\ve b):=\ve b_1\lor \ve b_3,
\end{align*}
% (operations in the finite field $\mathbb F_2$)
% \begin{align*}
%     f_{\mathrm{drive}}(\ve b):=\ve b_1 + \ve b_3 + \ve b_1 \times \ve b_3,
% \end{align*}
where $0$ means ``false'', $1$ means ``true'', and $\lor$ is the logical OR operator. The function  $f_{\mathrm{drive}}$ equals $1$ (signals an alarm) whenever either an over--voltage ($\ve b_1=1$) or an over--temperature ($\ve b_3=1$) event occurs. 
In \textit{charging mode}, the controller must determine whether charging is not permissible, 
represented by
\begin{align*}
    %f_{\mathrm{charge}}(\ve b): =(1+\ve b_1)\times(1+\ve b_2\times \ve b_3)
    f_{\mathrm{charge}}(\ve b): =\ve b_1\lor (\ve b_2\land \ve b_3),
\end{align*}
where $\land$ is the logical AND operator. The function $f_{\mathrm{charge}}$ equals $1$ (signals an alarm)  when either over--voltage $(\ve b_1=1)$, or the abnormal 
``under--voltage and hot'' ($\ve b_2=1$ and $\ve b_3=1$) condition occurs.

The Boolean Function Computation (BFC) problem is an abstraction of this example, where the receiver aims to recover one of several distinct Boolean functions of the same underlying binary sequence. To be precise, consider a point-to-point communication system where the transmitter chooses a message  represented by binary sequences of length $m$, and the decoder chooses a Boolean function $f: \{0,1\}^m\rightarrow \{0,1\}$ from a family $\mc F$. Here we assume that the set $\mc F$ is known to both the encoder and the decoder, but the function $f$ (chosen by the decoder) is unknown to the encoder. Let $i\in\{0,1\}^m$ denote the message at the transmitter, which is encoded to a codeword  $X^n $ of length $n$ and is transmitted through a noisy channel. With the channel output $Y^n$, the decoder is required to reliably decide if $f(i)=0$ or $f(i)=1$.

The BFC problem is a generalization of the \textit{identification via channel} problem studied in a seminal paper by Ahlswede and Dueck~\cite{ahlswede_identification_1989}.  The goal of identification is for the decoder to  decide if the transmitted message is  \textit{a particular one}. It is easy to see that the identification problem is equivalent to an instance of BFC problem with a suitable choice of function class $\mc F$ (details in Section \ref{sec:example_BFC}). After the appearance of~\cite{ahlswede_identification_1989} (and \cite{identification_feedback_1989} where feedback is considered),  many works have extended the identification problem to other set-ups. For point-to-point channels, a line of more recent work including \cite{salariseddigh_deterministic_2022} \cite{colomer_deterministic_2025}, investigates the identification problem when the encoders are restricted to a deterministic function of the message, for which the original work~\cite{ahlswede_identification_1989} only gives a very brief treatment. Explicit code constructions for identification are investigated in   \cite{verdu_explicit_1993} \cite{ahlswede_identification_1991} and further improved in  \cite{kurosawa_strongly_1999} and  \cite{gunlu_code_2022}. Extending the results to multi-user scenarios is considered in \cite{ahlswede_general_2008} for multiple-access channels and \cite{bracher_identification_2017} for broadcast channels. The readers are referred to the recent survey \cite{von_lengerke_codes_2025} for a more comprehensive summary of known results. 

Particularly relevant to our work is \cite{ahlswede_general_2008}, where a general information transfer problem is formulated. While a specialization of this model is equivalent to the Boolean function computation problem, only partial results are obtained in \cite{ahlswede_general_2008}. We make the connection and comparison explicit in Section \ref{sec:Ahlswede_models}.

The main contribution of this paper is to characterize the possible scaling law of $m$ and $n$ for the BFC problem. We show that the scaling  behavior heavily depends on the possible functions in $\mc F$ the receiver aims to decode.  Specifically, we consider $\mc F$ consisting of functions $f$ that have a \textit{Hamming weight} equal to or smaller than a certain threshold (denoted by $S$), where the Hamming weight is defined to be the cardinality of the pre-image of $1$ under $f$. With this characterization, we are able to give explicit results on the scaling behavior of $m$ with respect to $n$ for different values of $S$. Roughly speaking, $m$ scales exponentially with $n$ if $S$ is small (as in the identification problem, which corresponds to $S=1$), while it scales linear with $n$ when $S$ is large (the same behavior as in Shannon's transmission problem). When varying $S$ from large to small values, $m$ can scale linearly, quasi-linearly, polynomially, sub-exponentially, and exponentially with $n$, as shown in Table \ref{table:summary}. These results are obtained by establishing the achievability and converse results on the computation rate in Theorem \ref{thm:achievability} and Theorem \ref{thm:converse}, respectively. While the achievability and converse results coincide exactly only for special cases, they are tight in the sense of the scaling behavior for all cases of $\mc F$ considered in this paper.

Given a positive integer $M$, we use $\intset{M}$ to denote the set of integers $\{1,\ldots, M\}$. For a given set $\mc X$, we use $|\mc X|$ to denote the cardinality of the set.

\section{Problem formulation}
We consider a channel with the input alphabet $\mc X$ and the output alphabet $\mc Y$. A channel $W(\cdot|x^n)$ is viewed as a conditional distribution (transition kernel) over $\mathcal Y^n$, given the input codeword $x^n\in\mathcal X^n$.  We use $\bfm$ to denote the set of Boolean functions of $m$ inputs, $\bfm:=\{f:\{0,1\}^m\rightarrow \{0,1\} \}$. Consider the communication problem where a \textit{binary} message, represented by a binary sequence of length $m$, is chosen and transmitted (after encoding) by the transmitter. A Boolean function (unknown to the transmitter) is chosen by the receiver, and the goal is to reliably recover the function value of the transmitted message. The notion of Boolean function computation (BFC) code is formally defined as follows.

\begin{definition}[Boolean function computation (BFC) code]
Let  $m, n$ be positive natural numbers. An $(n,m,\mc F, \lambda_1, \lambda_2)$ Boolean function computation (BFC) code consists of encoders $(Q_i)_{i \in \{0,1\}^m}$ and decoding sets $(D_j)_{j \in \intset{|\mc F|}}$, such that
%is a collection $\{ Q_i, D_j: i\in\{0,1\}^m, j\in\intset{|\mc F|} \}$, such that
\begin{itemize}
\item $\mathcal F\subseteq \bfm$ is a set of Boolean functions with $m$ inputs
\item $Q_i$ is a probability distribution on $\mathcal X^n$ for all $i\in\{0,1\}^m$
\item $D_j\subset \mathcal Y^n$ for all $j\in \intset{|\mc F|}$
\item  false negative error: $\ve 1_{f_j(i)=1}Q_iW(D_j^c)\leq \lambda_1$ for all $i, j$
\item  false positive error: $\ve 1_{f_j(i)=0}Q_iW(D_j)\leq \lambda_2$ for all $i, j$
\end{itemize}
where we define  $QW(D):=\int W(D|x)Q(dx)$.
\end{definition}

Let $f^{-1}[1]$ denote the preimage of $1$ under $f$. The cardinality of the preimage $f^{-1}[1]$ is called the Hamming weight of the Boolean function $f$. We consider the class of  \textit{constant weight Boolean functions} defined as
\begin{align*}
\wcf(S) := \{f\in \bfm: |f^{-1}[1]|= S\}
\end{align*}
where  $S$ can take integer values between $0$ and $2^m$, and 
%\begin{align*}
$|\wcf(S)|= {2^m\choose S}$.
%\end{align*}
Similarly, we define the set of functions whose Hamming weight is \textit{smaller or equal to} $S$, and \textit{larger or equal to} $S$ as
\begin{align*}
\wcf(\leq S):= \{f\in \bfm: |f^{-1}[1]|\leq  S\}\\
\wcf(\geq S):= \{f\in \bfm: |f^{-1}[1]|\geq  S\}
\end{align*}

As we will see in the sequel, depending on the relationship between  $S$  and $m$, the number of messages that can be reliably computed via the channel exhibits different scaling behaviour with respect to $n$. To this end, we define the achievable computation rate as follows.

%\begin{definition}[Achievable computation rate for Boolean functions]
%We say   $R$ is \textit{an achievable computation rate with  rate function $L(n,R)$} for weight-constrained Boolean functions, if for all $\mathcal F\in \wcf (L(n,R), S)$, all $\lambda_1,\lambda_2>0$, and sufficiently large $n$, there exists an corresponding  $(L(n,R), n, \mc F,  \lambda_1,\lambda_2)$ BFC codes.
%\end{definition}

% \begin{definition}[Achievable computation rate for $\wcf(S)$]
%Given a \textit{rate function} $L:\mathbb N\rightarrow \mathbb R^+$,  we say  $R$ is \textit{an achievable computation rate} for the class $\wcf(S)$ with the rate function $L$, if for every $\lambda_1,\lambda_2>0$, every $\gamma>0$ and all sufficiently large $n$, there exists an   $(m, n, \wcf(S),  \lambda_1,\lambda_2)$ BFC code satisfying
%\begin{align*}
%\frac{1}{n} L(2^m) > R-\gamma
%\end{align*}
%\end{definition}

 \begin{definition}[Achievable computation rate]\label{def:rate}
Given a \textit{scaling function} $L:\mathbb R\times \mathbb N\rightarrow \mathbb R^+$ and a set of Boolean functions $\mc F$, we say  $R$ is \textit{an achievable computation rate} for   $\mc F$ with the scaling function $L$, if for all $\lambda_1,\lambda_2>0$, all  $\gamma>0$ and all sufficiently large $n$, there exists an   $(m, n, \mc F,  \lambda_1,\lambda_2)$ BFC code satisfying
\begin{align*}
m > L(R-\gamma, n).
%\frac{1}{n} L(2^m) > R-\gamma
\end{align*}
\end{definition}

To relate this to more familiar notions, the choice $L(R, n) = Rn$ corresponds to the case when $m$ scales linearly with $n$ (as in the classical transmission problem where the number of messages $2^m$ scales as $2^n$). The choice $L(R,n)=2^{Rn}$ corresponds to the case when $m$ scales exponentially with $n$ (as in the identification problem). We will see other scaling functions in Theorem \ref{thm:achievability} and \ref{thm:converse}. We point out that the scaling behaviour of $m$ (w.r.t $n$) is determined by the \textit{form} of the scaling function, whereas the rate $R$ is the coefficient at the front of $n$ in the scaling function. For example,  $L(R,n)=Rn, L(R, n)=Rn^2$ and $L(R,n)=2^{Rn}$ indicates that $m$ scales linearly, quadratically, and exponentially with $n$, respectively.

%We will consider different scaling functions in the paper, including $L(R,n)= 2^{Rn},  L(R, n)= Rn$

\begin{definition}[Computation capacity]
For a given scaling function  $L$ and a set of functions $\mc F$, the supremum of the achievable computation rate for $\mc F$, denoted by $\cbfc$,  is called the (Boolean function) computation capacity of the channel for $\mc F$ with the scaling function $L$.
\end{definition}

Before presenting the main results, we first give examples of the function class $\wcf(S)$ for different choices of $S$, and comment on connections to semantic evaluations in propositional logic, and other existing problem formulations in the literature.

\subsection{Examples: $\wcf(S)$ and $\wcf(\leq S)$  for different $S$}
\label{sec:example_BFC}

In this subsection, we use the notation $\ve b =(\ve b_1,\ve b_2,\ldots, \ve b_m)$ where each $\ve b_i$ is a Boolean variable.

%\begin{example}[$\wcf(1)$] Functions in the class  $\wcf(1)$ can be expressed explicitly as
%\begin{align*}
%f_j^{id}(\ve b) := \prod_{i\in A_j}\ve b_i \cdot \prod_{i\in A_j^c}(1-\ve b_i)
%\end{align*}
%for all subsets $A_j\subseteq [m], j=1,\ldots, 2^m$. Computing $\wcf(a)$ is equivalent to  the identification via channel problem. Indeed, $f_j(\ve b)=1$ if and only if $\ve b$ is the unique message that corresponds to the set $A_j$.
%\end{example}

\begin{example}[$\wcf(c)$]\label{example:S_constant} Let $c$ be a constant not depending  on $m$. The function in the class $\wcf(c)$ evaluates to $1$ only if the input message $\ve b$ comes from a set of size $c$. A special case is  $c=1$, where the functions in the class  $\wcf(1)$ can be expressed explicitly as
\begin{align}
f_j^{id}(\ve b) := \prod_{i\in A_j}\ve b_i \cdot \prod_{i\in A_j^c}(1-\ve b_i)
\label{eq:f_id}
\end{align}
for all subsets $A_j\subseteq \intset{m}, j=1,\ldots, 2^m$. Computing $\wcf(1)$ is equivalent to  the  \textit{identification via channel} problem \cite{ahlswede_identification_1989} (more details in Section \ref{sec:Ahlswede_models}). Indeed, $f_j^{id}(\ve b)=1$ if and only if $\ve b$ is the unique message that corresponds to the set $A_j$ (and there are in total $2^m$ of them).
%In this case we have $|\wcf(c)| = O( 2^{cm})$. The number of functions is a polynomial of order $c$ of the number of messages.
\end{example}

\begin{example}[$\wcf(\leq cm^{\beta})$] Let $c$ and $\beta>0$ be some constant not depending on $m$. If we ask the question: ``Are there exactly $\beta$ inputs equal to $1$"? This question can be formulated by the Boolean function
\begin{align*}
h_\beta(\ve b) := \ve 1_{\sum_{i=1}^m \ve b_i=\beta}
\end{align*}
The weight of this function is ${m\choose \beta}$ which is upper and lower bounded as $\frac{m^\beta}{\beta^\beta}\leq {m\choose \beta}\leq \frac{m^\beta}{\beta!}$. Hence this function belong to the class $\wcf(\leq cm^{\beta})$ for the choice $c=1/\beta!$.  Similarly, the question ``are there at most $\beta$ inputs equal to $1$" corresponds to the Boolean function
\begin{align*}
\tilde h_{\beta}(\ve b):=\ve 1_{\sum_{i=1}^m \ve b_i\leq \beta}
\end{align*}
These functions are known as the \textit{threshold function}, whose weight is upper bounded as $\sum_{i=1}^\beta {m\choose i}\leq (\frac{e}{\beta})^\beta m^\beta$. They belong to the set $\wcf(\leq cm^{\beta})$ with $c=(\frac{e}{\beta})^\beta$.
%In this case we have $|\wcf(cm^{\beta})| = O( 2^{cm})$
\end{example}

\begin{example}[$\wcf(c2^{\gamma m})$]\label{example:t_bit}  Let $c$ and $\gamma\in(0,1]$ be some constant not depending on $m$. If we want to identify the $t$-th bit of the message, we can use the function
\begin{align}
%f_t^{bit}(\ve b) := \ve b_t\cdot 1^{\prod_{i\neq t}\ve b_i}=\ve b_t
f_t^{bit}(\ve b) := \ve b_t
\label{eq:f_tbit}
\end{align}
It is easy to see that the set of functions $f_t^{bit}, t=1,\ldots, m$ belongs to the  class $\wcf(\frac{1}{2}2^{m})$. As another example, consider the function 
\begin{align*}
f_{S_k}^{\text{AND}}(\ve b):= \prod_{i\in S_k}\ve b_i
\end{align*}
where $S_k\subseteq\{1,\ldots, m\}$ with cardinality $k$. This function performs the logic AND operation on a subset (of size $k$) of inputs. The weight of this class of function is $2^{m-k}$. When $k \geq (1-\gamma)m$, they belong to the class $\wcf(\leq 2^{\gamma m})$.
\label{example:f_and}
\end{example}

\begin{example}[ranking]\label{example:ranking}
Let $int:\{0,1\}^m\rightarrow \mathbb N$ be the function that maps a binary sequence to an integer, where the sequence represents the binary expansion of that integer. In other words
\begin{align}
int(\ve b):= \sum_{i=1}^m b_i2^{m-i}
\label{eq:int}
\end{align}
Then we can define the ranking function
\begin{align}
f_{r}^{\text{rank}}(\ve b) := \ve 1_{int(\ve b)\leq r}
\label{eq:f_order}
\end{align}
for $r=0,\ldots, 2^m-1$. It is clear that the set of functions $f_r^{\text{rank}}$ belongs to the set $\wcf(r+1)$.
\end{example}

\subsection{Formulae in propositional logic}
The notion \textit{semantic communication} roughly refers to the communication problems where only the ``meaning", or semantics of messages is to be communicated, rather than the entirety of the message. Semantics exist in the context of languages,  with  \textit{propositional logic} being perhaps the simplest form, where there are only two semantic values, $\mathsf{True}$ (denoted by $\lt$) and $\mathsf{False}$ (denoted by $\lf$). In this subsection we comment on the connection between Boolean function computation and the semantic evaluation in propositional logic.

In the setup of propositional logic \cite{enderton_mathematical_1972},  $\mc P=\{p_1,\dots,p_m\}$ denotes a finite set of \emph{propositional atoms}. The set of \emph{well-formed formulae} \(\mc L(\mathcal{P})\) is defined recursively as the smallest set of finitely long tuples of elements of $\mc P \cup \{\neg, \lor, \land, (, )\}$ with the properties:
\begin{itemize}
    \item for every $p \in \mc P$, we have $p \in \mc L(\mathcal{P})$,
    \item if $\varphi \in \mc L(\mathcal{P})$, then $\neg (\varphi) \in \mc L(\mathcal{P})$,
    \item if $\varphi \in \mc L(\mathcal{P})$ and $\psi \in \mc L(\mathcal{P})$, then $(\varphi) \lor (\psi) \in \mc L(\mathcal{P})$ and $(\varphi) \land (\psi) \in \mc L(\mathcal{P})$,
\end{itemize}
where \(p_i\in\mathcal{P}\), and \(\neg\), \(\land\), \(\lor\) are the logical connectives \emph{negation}, \emph{conjunction}, and \emph{disjunction}. For convenience, we do not always write out all the parentheses of every formula (when the order of operations is clear by common conventions), and we also define the connectives
\begin{itemize}
    \item $\varphi \to \psi := \neg \varphi \lor \psi$ \emph{(implication)}
    \item $\varphi \leftrightarrow \psi := (\varphi \land \psi) \lor (\neg \varphi \land \neg \psi)$ \emph{(equivalence)}
    \item $\varphi \oplus \psi := (\varphi \land \neg \psi) \lor (\neg \varphi \land \psi)$ \emph{(exclusive or)}.
\end{itemize}

A \emph{truth assignment} is a function \(\tau:\mc P\to\{\lt,\lf\}\) assigning a truth value to each propositional atom.
The \emph{semantics} of formulae is given by the evaluation map \(v_\tau : \mc L(\mathcal{P})\to\{\lt, \lf\}\), defined recursively via
\begin{align*}
v_\tau(p_i) &= \tau(p_i)\\
v_\tau(\neg \varphi)&=\begin{cases}
\lt \quad \text{if } v_\tau(\varphi) = \lt\\
\lf\quad \text{otherwise}
\end{cases} \\
v_\tau(\varphi\land\psi) &=\begin{cases}
\lt\quad \text{if } v_\tau(\varphi) = \lt \text{ and }  v_\tau(\psi)= \lt\\
\lf \quad \text{otherwise}
\end{cases}\\
v_\tau(\varphi\lor\psi) &=\begin{cases}
\lt\quad \text{if } v_\tau(\varphi) = \lt \text{ or }  v_\tau(\psi)= \lt\\
\lf \quad \text{otherwise.}
\end{cases}
\end{align*}
In this way, each formula \(\varphi\) defines a Boolean-valued function of its atomic propositions under the truth value assignment $\tau$.   We call two formulae $\varphi_1$ and $\varphi_2$ \textit{tautologically equivalent}, if they have the same truth value on all truth assignments of the atoms, namely $\forall \tau, v_{\tau}(\varphi_1)=v_{\tau}(\varphi_2)$. Otherwise we call them \textit{tautologically distinct}. With $m$ atoms, there are infinitely many formulae, but there only exist $2^{2^m}$ tautologically distinct formulae, with $2^m$ being the number of possible truth value assignments.

The key observation is that Boolean function computation can be viewed as evaluating the semantics of a chosen formula by the receiver, where the truth assignment is determined by the transmitter. Specifically, let $\ve b = (\ve b_1,\ldots \ve b_m)\in\{0,1\}^m$ denote the message at the transmitter, and we associate each coordinate \(b_i\) with an atom \(p_i\in\mathcal{P}\) and define the truth assignment \(\tau_{\ve b}(p_i)=\lt\) if \(\ve b_i=1\) and \(\tau_{\ve b}(p_i)=\lf\) if \(\ve b_i=0\). Any Boolean function \(f\in\mc F_m\) can then be represented by a propositional formula \(\varphi_f\) such that for every input \(\ve b\in\{0,1\}^m\),
\begin{align*}
f(\ve b)=1 \;\;\Longleftrightarrow\;\; v_{\tau_{\ve b}}(\varphi_f)=\lt
\end{align*}
Thus, evaluating the Boolean function \(f\) on the message $\ve b$ is equivalent to evaluating the propositional formula \(\varphi_f\) under the truth assignment \(\tau_{\ve b}\). Furthermore, there is a one-to-one mapping between the $2^{2^m}$ different Boolean functions and $2^{2^m}$ tautologically distinct formulae. The logic \text{AND} function defined in Example \ref{example:f_and} gave an example of this connection. As another simple example, the formula  that calculates the parity of three atoms
\begin{align*}
\varphi_{\text{parity}} := p_1 \oplus p_2 \oplus p_3
\end{align*}
corresponds to the Boolean function
\begin{align*}
f_{\varphi_{\text{parity}} }(\ve b) := (\ve b_1+\ve b_2+\ve b_3)\pmod 2.
\end{align*}

\begin{example}[Hamming weight of DNF]
Disjunctive normal forms (DNF) are a canonical form of formulae consisting of a disjunction of one or more conjunctions of one or more literals (an atomic or its negation). For example, $(p_1\land p_2\land \neg p_3)\lor p_4$ is in DNF where as $\neg(p_1\land p_2)$ is not. It is known that \cite[Coro. 15C]{enderton_mathematical_1972} in propositional logic, any well-formed formula in $\mc L(\mc P)$ is tautologically equivalent to a formula in DNF. If we have a DNF with $t$ disjoint conjunction terms (i.e. no two conjunction terms contain the same atoms), and each conjunction contains $k$ atoms, then the Hamming weight of this DNF is $t2^{m-k}$. In general if there are overlapping atoms in different conjunction terms, the Hamming weight of this DNF is smaller than or equal to $t2^{m-k_{\text{min}}}$ where $k_{\text{min}}$ is the number of atoms in the shortest conjunction term.
\end{example}

\subsection{Connection to identification via channels and other Ahlswede models in \cite{ahlswede_general_2008}}
\label{sec:Ahlswede_models}

The Boolean function computation problem is a generalization of the identification via channels problem. A $(n, N, \lambda_1, \lambda_2)$ identification (ID) code is a collection $\{(Q_i, D_i), i=1,\ldots, N\}$ with probability distributions $Q_i$ on $\mc X^n$ and $D_i\in\mc Y^n$ such that the two types of error probability satisfy
\begin{align*}
&Q_iW(D_i^c)\leq \lambda_1, i=1,\ldots, N\\
&Q_iW(D_j)\leq \lambda_2, i, j =1,\ldots, N, i\neq j
\end{align*}
As discussed in Example \ref{example:S_constant} above, an $(n, N, \lambda_1, \lambda_2)$ ID code (assuming $N$ is a power of $2$ for simplicity) is equivalent to an $(n, \log N, \mc F, \lambda_1,\lambda_2)$ BFC code where  $\mc F$ contains the functions $f_j^{id}, j=1,\ldots, 2^m$ defined in  (\ref{eq:f_id}).

In \cite{ahlswede_general_2008}, Ahlswede proposed several models which generalize the identification via channel problem.  In all models,  there is a family of partitions $\Pi = \{\pi_j\}_j$ where each   $\pi_j=\{P_{j1},\ldots, P_{jr}\}$ is a partition of the messages $\mc M$, where we call $P_{ji}$ a member of the partition $\pi_j$. The communication problem is described as follows:
\begin{itemize}
\item the transmitter picks a message $m$ from $\mc M$, and transmits a codeword associated to the message $m$
\item upon receiving the channel output,  for any partition $\pi_j$ from the family $\Pi$, the decoder needs to decide which member of the partition $\pi_j$ contains the true message $m$
\end{itemize}
Identifying the set of messages $\mc M$ with $\{0,1\}^m$, our model is a special case of the above model, where each partition $\pi_j$ corresponds to a Boolean function $f_j$, which partitions $\mc M$ into two (i.e. $r=2$) disjoint sets, corresponding to $f^{-1}[1]$ and $f^{-1}[0]$, respectively. In particular,  the  $K$-identification problem~\cite{ahlswede_general_2008} (see Model 3 below)  is identical to our problem with the choice of functions $\wcf(K)$. We repeat the communication models discussed in~\cite{ahlswede_general_2008} (in its original numeration) below, and connect them to the BFC problem.
\begin{itemize}
\item Model 2 (identification): $\Pi_I=\{\pi_j: \pi_j=\{ \{j\}, \mc M\backslash \{j\}\}, j\in \mc M\}$. This corresponds to our problem with the choice  of functions to be $f_j^{id}$ defined in \eqref{eq:f_id}. Notice that we have
\begin{align*}
\{j\} = (f_j^{id})^{-1}[1]
\end{align*}
\item Model 3 ($K$-identification):  $\Pi_K=\{\pi_{\mc S}: \pi_{\mc S}=\{\mc S, \mc M\backslash S\}, |\mc S|= K, \mc S\subseteq \mc M\}$. This corresponds to our problem with the choice  of functions to be the set $\wcf(K)$. Here we have
\begin{align*}
\mc S = f^{-1}[1] \text{ for  } f\in \wcf(K)
\end{align*}
\item Model 4 (ranking): $\Pi_R = \{\pi_r: \pi_r = \{\{1,\ldots, r\}, \{r+1,\ldots, 2^m\}\}, 1\leq r\leq 2^m\}$. This corresponds to our problem with the choice of functions $f_r^{\text{rank}}$ defined in \eqref{eq:f_order}. Notice that
\begin{align*}
\{1,\ldots, r\} = (f_r^{\text{rank}})^{-1}[1]
\end{align*}
\item Model 5 (general binary questions) : $\Pi_B = \{\pi_{\mc A}: \pi_{\mc A}=\{\mc A, \mc M\backslash \mc A\}, \mc A\subset \mc M\}$. This corresponds to our problem with the set of functions $\{f_{\mc A}, \mc A\subset \mc M\}$, where $f_{\mc A}$ is the function whose preimage under $1$ is the set $\mc A$.
\item Model 6 ($t$-th bit): $\Pi_C = \{\pi_t: \pi_t =\{  \{\ve b\in\{0,1\}^m: \ve b_t=1\}, \{\ve b\in\{0,1\}^m: \ve b_t=0\}\} , t=1,\ldots, m \}$. This corresponds to our problem with the choice of functions to be $f_t^{bit}$ defined in \eqref{eq:f_tbit}. Notice that
\begin{align*}
 \{\ve b: \ve b_t=1\} = (f_t^{bit})^{-1}[1]
\end{align*}
\end{itemize}

It may be worth pointing out that the Shannon's original transmission problem (recover the transmitted message) is not covered in the BFC framework, but can be formulated in Ahlswede's partition framework (see Model 1 in \cite{ahlswede_general_2008}). Nevertheless, it can be formulated as a \textit{multiple Boolean functions computation} problem.  For example, if we extend Model 6 above, asking the the receiver to simultaneously compute $m$ functions $f_1^{bit},\ldots, f_m^{bit}$, then we effectively recover the original message. We only consider computing a single function in this paper, and leave multiple function computation for future work.

\section{Main result}
We state the main results of our paper in this section.

\begin{theorem}[Achievability]
Assume the channel $W(\cdot | x^n)$ has a Shannon capacity $C$, and let $c>0$ be a constant independent from $m$ and $n$. Then,
\begin{itemize}
\item (1) for the function class $\wcf(\leq c)$ where $c$ is a positive integer,   a computation rate of $C$ is  achievable  with the scaling function $L(R,n):=2^{Rn}$
\item (2) for the function class $\wcf(\leq cm^{\beta})$ with $\beta>0$,  a computation rate of $\frac{1}{1+2\beta} C$ is  achievable   with the scaling function $L(R,n):=2^{Rn}$
\item (3) for the function class $\wcf(\leq cm^{\log m})$, a computation rate of  $\sqrt{C/2}$  is achievable with  the scaling function $L(R,n):= 2^{R\sqrt{n}}$
%\item the rate $C$ is an achievable computation rate for $\mc F_m^{CW}(c2^m)$ with the rate function $\log$, for constant $c>0$ not depending on $m,n$;
\item (4) for the function class $\wcf(\leq c2^{m^{1/b}})$ with $b>1$, a computation rate of $\frac{C}{3(\lfloor 3/C+1\rfloor)^{b-1}}$ is achievable\footnote{The expression simplifies to $C/3$ for $C>3$.}  with the scaling function $L(R,n):=Rn^{b}$.  
\item (5) for the function class $\wcf(\leq c2^{m/\log m})$, a computation rate of $C/2$ is achievable with the scaling function $L(R,n):=Rn\log  n$
\item (6) for the  function class $\wcf(\leq c2^{\gamma m})$ with $\gamma\in(0,1]$, a computation rate of $C$ is achievable with the scaling function $L(R,n):=Rn$
\end{itemize}
\label{thm:achievability}
\end{theorem}

The results are summarized  in Column 2 in Table \ref{table:summary}.

\begin{theorem}[Converse]
Assume the channel $W(\cdot | x^n)$ has a Shannon capacity $C$.  Let  $\cbfc$ denote the computation capacity for $\wcf(S)$ with a  scaling function specified below. We have the following statements:
\begin{itemize}
\item (1) for $S=c$ (a positive integer), $\cbfc \leq C$ with the scaling function $L(R,n)=2^{Rn}$
\item (2) for $S=cm^{\beta}$ with $\beta>0$, $\cbfc \leq \frac{C}{1+\beta}$ with the scaling function $L(R,n)=2^{Rn}$
\item (3) for $S=cm^{\log m}$, $\cbfc \leq \sqrt{C}$ with the scaling function $L(R,n)=2^{R\sqrt{n}}$
\item (4) for $S=c2^{m^{1/b}}$ with $b> 1$, $\cbfc \leq C^b$ with the scaling function $L(R,n)=Rn^b$
\item (5) for $S=c2^{m/\log m}$, $\cbfc\leq C$ with the scaling function $L(R,n)=Rn\log n$
\item (6) for $S=c2^{\gamma m}$ with $\gamma\in(0,1)$,$\cbfc \leq \frac{C}{\gamma}$ with the scaling function $L(R,n)=Rn$
\end{itemize}
\label{thm:converse}
\end{theorem}

\begin{table*}[tbh!]
\centering
\begin{tabular}{l c c c}
\toprule
Hamming weight of $f\in\mc F$  &  \makecell{Achievability  \\ (lower bound on $m$)} & \makecell{Converse \\ (upper bound on $m$)} &  Asymptotic expression  of $m$ \\
\midrule
(1) \quad $c$ (constant) & $2^{Cn}$ (\cite{ahlswede_identification_1989}  for $c=1$) & $2^{Cn}$(\cite{ahlswede_identification_1989} for $c=1$)  & $\Theta(2^n)$, exponential  \\\midrule
(2) \quad $cm^{\beta}, \beta>0$  & $2^{\frac{C}{1+2\beta}n}$ (\cite{ahlswede_general_2008}) for $c=1$) & $2^{\frac{C}{1+\beta}n}$ (\cite{ahlswede_general_2008} for $c=1$) & $\Theta(2^n)$, exponential \\
\midrule
(3) \quad $cm^{\log m}$  & $2^{\sqrt{Cn/2}}$ & $2^{\sqrt{Cn}}$ &$\Theta\left(2^{\sqrt{n}}\right) $, sub-exponential\\
\midrule
(4) \quad $c2^{m^{1/b}}, b>1$ &  $\frac{C}{3(\lfloor 3/C+1\rfloor)^{b-1}}n^b$ & $(Cn)^{b}$ & $\Theta(n^{b})$, polynomial \\
\midrule
(5) \quad $c2^{m/\log m}$ & $\frac{C}{2}n\log n$  & $C n\log n$ & $\Theta(n\log n)$, quasi-linear \\
\midrule
(6) \quad $c2^{\gamma m}, \gamma\in(0,1)$ & $Cn$ & $\frac{C}{\gamma}n$ & $\Theta(n)$, linear \\
\midrule
Shannon problem & $Cn$ & $Cn$ & $\Theta(n)$\\
\bottomrule
\end{tabular}
\caption{A summary of the achievability (Theorem~\ref{thm:achievability}) and converse  (Theorem~\ref{thm:converse}) results. Here $m$ denotes the length of the message and $n$ is the channel use. For all six cases considered here, the upper and lower bounds coincide in the sense of the scaling behavior of $m$. In the last row we include the Shannon problem for reference.}%The identification problem corresponds to the first row with the choice $c=1$. 
\label{table:summary}
\end{table*}

\begin{remark}
The converse results in the above theorem hold for the function class $\wcf(S)$, the class of Boolean functions whose weight is \textit{equal to} $S$. It follows directly that the stated results also hold for the class $\wcf(\leq S)$ as it is a larger class that contains $\wcf(S)$. Therefore the above theorem also serves as the converse counterpart of Theorem \ref{thm:achievability}. Vice versa, the achievability results in Theorem \ref{thm:achievability}, stated for the function class of the form $\wcf (\leq S)$, also hold for the class $\wcf(S)$.
\end{remark}

The results in  Theorem \ref{thm:converse} are summarized in Column 3 in Table \ref{table:summary}. Comparing it with the achievability results in Theorem \ref{thm:achievability}, we see that the results exactly match (including the leading coefficient) in Case 1) for function class $\wcf(\leq c)$.  For all other cases, the achievability and the converse results match in the scaling law sense, shown in the last column in Table \ref{table:summary}.  In the last row of the table we include the result of  Shannon's transmission problem as a reference, though noting that it is not an instance of the BFC problem.

%[Q: will the result be the same if require the weights to be smaller or equal to $S$ instead of equal to $S$?]

We have the following remarks:
\begin{itemize}
\item The results for Cases (3), (4), and (5) are new. Results related to Cases 1), 2) already appeared in the literature, as discussed below.
\item As discussed in Section \ref{sec:Ahlswede_models}, in Case (1) with the choice $c=1$, the BCF problem with $\wcf(\leq 1)=\wcf(1)$ is equivalent to the identification via channels problem studied in \cite{ahlswede_identification_1989} where the double exponential capacity result was established.  The same result holds if the Hamming weight is a constant larger than $1$.
\item For Case (2), the choice $S=m^{\beta}$ has been studied in \cite[Section 3]{ahlswede_general_2008} under the name $K$-identification (see Model 3 in Section \ref{sec:Ahlswede_models}) where the Hamming weight $S$ (denoted by $K$ in \cite{ahlswede_general_2008}) is parametrized as $2^{\kappa  n}$ (i.e. by $n$ instead of $m$). It was shown that $R=C-2\kappa$ (corresponding to $m=2^{n(C-2\kappa)}$ in our case) is achievable and a converse result of $R\leq C-2\kappa$ (corresponding to $m\leq 2^{n(C-\kappa)}$ in our case) is established. Identifying $2^{\kappa n}$ with $m^{\beta}$, it can be checked that our result matches that in \cite{ahlswede_general_2008}. For example,  our achievable rate $m=2^{\frac{C}{1+2\beta}n}$ is matched with the result $2^{n(C-2\kappa)}$ in~\cite{ahlswede_general_2008} by letting $2^{\kappa n}=m^{\beta}$:
\begin{align*}
2^{n(C-2\kappa)}&=2^{nC}\cdot (2^{\kappa n})^{-2}= 2^{nC}\cdot (m^\beta)^{-2}\\
&=2^{nC}\cdot 2^{\frac{-2\beta }{1+2\beta}Cn}=2^{\frac{C}{1+2\beta}n}
\end{align*}
%\item Several communication problems are discussed in \cite[Section 4]{ahlswede_general_2008} whose capacity is equal to the ordinary Shannon capacity, namely, $m=Cn$ asymptotically. These problems include the ranking problem in Example \ref{example:ranking} (Model 4) and the $t$-th bit identification problem in Example \ref{example:t_bit} (Model 6). It is discussed in Section \ref{sec:example_BFC} that the former problem corresponds to a BFC problem with the set $\wcf(\leq 2^m)$ and the latter problem corresponds to a BFC problem with the set $\wcf(\frac{1}{2}2^m)$. Both examples correspond to  Case (6) of our results.
 \end{itemize}

%\begin{remark}
The above results stated with the characterization of the Boolean functions in terms of their preimage under $1$. One can equally choose to work with the preimage under $0$.   In particular, the set $\wcf(S)$, whose elements satisfy $|f^{-1}[1]|=S$ and $|f^{-1}[0]|=2^m-S$, essentially has the same property as the set $\wcf(2^m-S)$, by flipping the output of the functions in the former set. We formalize the observation in the following Lemma.
%\end{remark} 

\begin{lemma}
Given an $(n,m, \wcf(S), \lambda_1,\lambda_2)$ BCF code, we  can construct an $(n,m, \wcf(2^m-S), \lambda_2,\lambda_1)$ BCF code.
\end{lemma}
\begin{proof}
Given the encoder $Q_i, D_j$ of the $(n,m, \wcf(S), \lambda_1,\lambda_2)$ BCF code, and for each function $f_j\in\wcf(S)$, we define a new set of functions $\tilde f_j=1-f_j$,  and the corresponding encoders and decoders as
\begin{align*}
\tilde Q_i&:=Q_i\\
\tilde D_j&:=D_j^c
\end{align*}
Clearly, the set $\{\tilde f_j\}_j=\wcf(2^m-S)$, and it  holds $\tilde f_j^{-1}[1]= f_j^{-1}[0]$, $\tilde f_j^{-1}[0]= f_j^{-1}[1]$. To characterize the error probability of the new code, we have
\begin{align*}
\ve 1_{\tilde f_j(i)=1}Q_iW(\tilde D_j^c)=\ve 1_{f_j(i)=0}Q_iW(D_j)\leq \lambda_2\\
\ve 1_{\tilde f_j(i)=0}Q_iW(\tilde D_j)=\ve 1_{f_j(i)=1}Q_iW(D_j^c)\leq \lambda_1
\end{align*}
where the inequalities follow from the definition of the $(n,m, \wcf(S), \lambda_1,\lambda_2)$ BCF code.
\end{proof}

A similar result holds for the set of functions $\wcf(\leq S)$, whose proof is essentially the same as the proof of the above Lemma.
\begin{lemma}
Given an $(n,m, \wcf(\leq S), \lambda_1,\lambda_2)$ BCF code, we can construct an $(n,m, \wcf(\geq 2^m-S), \lambda_2,\lambda_1)$ BCF code.
\end{lemma}

%
%\begin{table}[h!]
%\centering
%\begin{tabular}{l c c c}
%\toprule
%Hamming weight &  achievable & converse  &  scaling law of $2^m$ \\
%\midrule
%constant & $2^{2^{Cn}}$ & $2^{2^{Cn}}$ & $\Theta(2^{2^n})$ \\\midrule
%$cm^{\beta}, \beta\geq 1$  & $2^{2^{\frac{C}{1+2\beta}n}}$ & $2^{2^{\frac{C}{1+\beta}n}}$ & $\Theta(2^{2^n})$ \\
%\midrule
%$c2^{m^b}, b\in(0,1]$ & $2^{\frac{2C}{3}n^{1/b}}$ & $2^{(Cn)^{1/b}}$ & $\Theta(2^{n^{1/b}})$ \\
%\midrule
%$c2^{\gamma m}, \gamma\in(0,1]$ & $2^{Cn}$ & $2^{\frac{C}{\gamma}n}$ & $\Theta(2^n)$ \\
%\bottomrule
%\end{tabular}
%\caption{Main results}
%\end{table}

\section{Proof of achievability}
\label{sec:proof_achievability}

In this section we prove the achievability result in Theorem \ref{thm:achievability}. Central to the proof is the following result, stated in~\cite[Proposition 1]{ahlswede_identification_1989}. The following proposition is a slightly modified version of \cite[Proposition 1]{ahlswede_identification_1989}. %The modification is needed to prove the result in Theorem \ref{thm:achievability}.

\begin{proposition}[Maximal code, modified from \cite{ahlswede_identification_1989}, Proposition 1]
\label{prop:maximal_code}
Let $\mc Z$ be a finite set with cardinality $|\mc Z|>6$  and let $\lambda\in(0,1/2)$ be given. Let $\epsilon<\frac{1}{6}$ and $\epsilon |\mc Z|\geq 1$. Then a family $A_1,\ldots, A_N$ of subsets of $\mc Z$ satisfying the following properties exists 
\begin{align*}
|A_i|&= M':=\lceil\epsilon |\mc Z|\rceil\geq 1, i=1,\ldots, N\\
|A_i \cap A_j|&<\lambda M', i,j=1,\ldots, N, i\neq j
\end{align*}
and
\begin{align*}
N\geq=H(\lambda,M')M'^{-1}\left( \frac{1-\epsilon}{2\epsilon}\right)^{\lceil\lambda M'\rceil}
\end{align*}
%\begin{align*}
%H(\lambda,M')M'^{-1}\left( \frac{1-\epsilon}{2\epsilon}\right)^{\lceil\lambda M'\rceil}\leq N\leq H(\lambda,M')M'^{-1}\left( \frac{2}{\epsilon}\right)^{\lceil\lambda M'\rceil}+1\\
%%&\geq M'^{-1} 2^{M'(\lambda \log(1/\epsilon-1)-H(\lambda, M'))}-1
%\end{align*}
where $H(\lambda,M'):={M'\choose \lceil\lambda M'\rceil}^{-1}$.  %In the special case when the first condition holds with $M'=1$ (and the second condition is equivalent to $A_i\cap A_j=\emptyset$), we have the exact characterization $N=|Z|$.
\end{proposition}

The above result gives a tighter lower bound on $N$ than that stated in \cite[Proposition 1]{ahlswede_identification_1989}, which is the essential difference that is needed to prove Theorem \ref{thm:achievability}.  The proof is almost the same as the proof of \cite[Proposition 1]{ahlswede_identification_1989}, and is included in the Appendix for completeness.

\textbf{Construction of BFC codes.} We use the fact that for a channel with the (Shannon) capacity $C>0$, for all  $\delta\in(0,1/2)$, $\xi>0$, and  all sufficiently large $n$, there exists an $n$-length transmission code $\{(\ve u_i, \mc E_i)| i=1,\ldots, M\}$, where all $\ve u_i\in\mathcal X^n$ are different, with the following property
\begin{align}
&\mc E_i\cap\mc E_j=\emptyset, i\neq j, \nonumber\\
&W(\mc E_i^c|\ve u_i)\leq \delta, \nonumber\\
&M= 2^{n(C-\xi)}.  \label{eq:good_channel_code}
\end{align}

Fix a small error probability $\delta\in(0,1/2)$,  let $\mc Z=\{\ve u_1,\ldots, \ve u_M\}$ be the set of codewords from a transmission code satisfying the properties \eqref{eq:good_channel_code}\footnote{Here we assume $2^{n(C-\xi)}$ is an integer for simplicity of the analysis. Without this assumption, we can assume that $M$ is some integer that satisfies $2^{n(C-\xi_1)}\leq M\leq 2^{n(C-\xi_2)}$ for arbitrarily small $\xi_1$ and $\xi_2$, and the same result can be proved by modifying the proof steps accordingly.}.   Then, there exist $N$ subsets of $\mc Z$ with the property specified in Proposition \ref{prop:maximal_code}. We use $A_i, i\in \{0,1\}^m$ (indexed by binary sequences of length $m$) to denote the first $2^m$ of the $N$ subsets, where $m$ is chosen to be an integer such that
\begin{align}
2^m\leq H(\lambda,M')M'^{-1}\left( \frac{1-\epsilon}{2\epsilon}\right)^{\lceil\lambda M'\rceil}=:\tilde N
\label{eq:m}
\end{align}
where $M'=\lceil\epsilon M\rceil$. Proposition \ref{prop:maximal_code} guarantees the existence of $\{A_i\}_{i\in\{0,1\}^m}$ if \eqref{eq:m} is satisfied. Also notice that $\epsilon$ should satisfy 
\begin{align}
\frac{1}{6}>\epsilon\geq  \frac{1}{|\mc Z|}=2^{-n(C-\xi)}
\label{eq:epsilon_const}
\end{align}
by the assumption in Proposition \ref{prop:maximal_code}. 

The encoder and the decoder of the BFC code are defined as follows.

\textit{Encoder:} Define for every $i\in\{0,1\}^m$, the stochastic encoder $Q_i$ to be the uniform distribution over $A_i$:
\begin{align}
Q_i(\ve x)=\ve 1_{\ve x\in A_i}\frac{1}{|A_i|} 
\label{eq:encoder}
\end{align}

\textit{Decoder: }The decoder set for the function $f_j$ is defined to be
\begin{align}
\mathcal D_j &:= \bigcup_{i\in f_j^{-1}[1]}\bigcup_{k:  \ve u_k\in A_i} \mc E_k\nonumber\\
&=\bigcup_{k: \ve u_k\in B_j}\mc E_k \label{eq:Dj_def}
\end{align}
where we define $B_j:=\bigcup_{\ell\in f_j^{-1}[1]}A_\ell$. 
%where $Z_j$ denotes the set of messages (binary sequences) with which $f_j$ evaluates to $1$, namely
%\begin{align*}
%Z_j = \{\ve b\in\{0,1\}^m: f_j(\v b)\}
%\end{align*}

\textbf{Error probability analysis.} The following lemma serves as an intermediate result to characterize the error probability of the above coding scheme.

\begin{lemma}
For all $\delta,\xi\in(0,1/2)$, for all sufficiently large $n$, for all  $\lambda\in (0,1/2)$, and all $m$ and $\epsilon$ such that \eqref{eq:m} and \eqref{eq:epsilon_const} hold, the  $(n,m,\mc F(\leq S), \lambda_1, \lambda_2)$  BCF code  constructed above satisfies 
\begin{align*}
\lambda_1=\delta, \lambda_2=S\lambda+\delta
\end{align*}
%for and all $\epsilon$ satisfying the condition .
\label{lemma:error_probability}
\end{lemma}
The lemma is proved in Appendix \ref{appendix:error_probability}. It shows that the false negative error $\lambda_1$ can be made arbitrarily small as the result holds for all $\delta\in(0,1/2)$.   In the rest of the proof, we upper bound the false positive error $\lambda_2=S\lambda+\delta$ for different values of $S$ and an appropriate choice of $\lambda$. 

%In all cases we will let $\lambda=\delta/S$ so that the false positive error is upper bounded by $2\delta$.

{\bf Case 1 (constant $S$).} We consider the case when $S$ is some positive integer not depending on $m$ and $n$. In this case we choose $m=\lceil2^{(C-\xi')n}\rceil$ for some $\gamma>\xi'>\xi$, with $\xi$ given in  \eqref{eq:good_channel_code} and $\gamma$  given as per Definition \ref{def:rate}.  We also choose $\lambda = \delta/S\in (0,1/2)$ and $\epsilon=\min\{\frac{1}{7},\frac{1}{2^{2S/\delta+1}+1}\}$.  It is easy to see that with this choice of $\epsilon$, \eqref{eq:epsilon_const} is satisfied for large enough $n$. 

To verify (\ref{eq:m}),  we first lower bound $\tilde N$ as
\begin{align*}
\tilde N&= H(\lambda,M')M'^{-1}\left( \frac{1-\epsilon}{2\epsilon}\right)^{\lceil\lambda M'\rceil} \nonumber\\
&\stackrel{(a)}{\geq}  M'^{-1} 2^{M'\lambda \log((1-\epsilon)/2\epsilon))}H(\lambda,M') \nonumber\\
&\stackrel{(b)}{\geq} \frac{1}{\epsilon M+1} 2^{M' (\lambda\log((1-\epsilon)/2\epsilon)-1) }\\
&=\frac{1}{\epsilon M+1} 2^{M' (\delta\log ((1-\epsilon)/2\epsilon)/S-1) }
\end{align*}
where $(a)$ uses ${\lceil\lambda M'\rceil}\geq \lambda M'$, and  $(b)$ uses the fact $M'\leq \epsilon M+1$  and
\begin{align*}
H(\lambda,M') = {M' \choose \lceil\lambda M'\rceil}^{-1} \geq 2^{-M'},
\end{align*}
due to the bound on the binomial coefficient ${n\choose k}\leq 2^n$. In the last step we substitute $\lambda = \delta/S$. It is easy to verify that with the $\epsilon$ chosen above, we have $\delta\log ((1-\epsilon)/2\epsilon)/S\geq 2$ because $\epsilon\leq \frac{1}{2^{2S/\delta+1}+1}$. We have the lower bound
\begin{align*}
\tilde N&\geq \frac{1}{\epsilon M+1} 2^{M'}\\
&\geq \frac{1}{M+1} 2^{\epsilon M}\\
&\geq \frac{1}{|\mc X|^n+1}2^{\epsilon  2^{n(C-\xi)}}
\end{align*}
In the last step we used the fact that the number of codewords $|M|$ is smaller than $|\mc X|^n$. 

With our choice of $m$,  \eqref{eq:m} is satisfied if we can show the following inequality holds 
%\begin{align*}
%2^m\leq \frac{1}{|\mc X|^n+1}2^{\epsilon  2^{n(C-\xi)}}
%\end{align*}
%for large enough $n$.  Substituting $m$ and taking logarithm on both sides yields
\begin{align*}
\lceil 2^{(C-\xi')n}\rceil\leq -\log(|\mc X|^n+1)+\epsilon 2^{n(C-\xi)}
\end{align*}
The above inequality holds if
\begin{align*}
2^{(C-\xi')n+1}\leq -\log(|\mc X|^n+1)+\epsilon 2^{n(C-\xi)}
\end{align*}
which is equivalent to
\begin{align*}
2^{-(\xi'-\xi)n+1}\leq (-\log (|\mc X^n|+1))2^{-n(C-\xi)}+\epsilon.
\end{align*}
Observe that as $n\rightarrow \infty$, the RHS of the above inequality approaches $\epsilon>0$ while the LHS approaches $0$ for all $\xi'>\xi$. Hence we have shown that $m=\lceil 2^{(C-\xi')n}\rceil$ satisfies \eqref{eq:m} for  $\xi'>\xi$ and sufficiently large $n$.  

With both \eqref{eq:m} and \eqref{eq:epsilon_const} verified, and the choice $\lambda=\delta/S$, the false positive error is upper bounded by $\lambda_2=S\lambda+\delta= 2\delta$. This shows that  $C$ is an achievable computation rate with the scaling function $L(R,n)=2^{Rn}$ according to Definition \ref{def:rate}.

{\bf Case 2 ($S=c m^\beta$).}  We consider the case when $S=cm^{\beta}$ for some constant $c,\beta>0$  not depending on $m, n$.   In this case we choose $\lambda = 2^{-\frac{\beta C}{1+2\beta}n}$, $\epsilon=\frac{\lambda}{4e}$, and $m=\left\lceil2^{(\frac{C}{1+2\beta}-\xi')n}\right\rceil$ for some $\xi'$ chosen as in Case 1.   It holds $\lambda<1/2$ for large enough $n$, and it is easy to see that \eqref{eq:epsilon_const} is also satisfied for large enough $n$.

%We choose $\lambda = 2^{-\frac{\beta C}{1+2\beta}n}$ and let $n$ be large enough such that $\lambda<1/2$.  Similarly, we show  for this case that we can choose $m=\lceil2^{(\frac{C}{1+2\beta}-\xi')n}\rceil$ for an arbitrarily small $\xi'$ such that 1) the inequality \eqref{eq:m} is satisfied, and 2)  the false positive error in \eqref{eq:lambda_2}  is small.

To verify \eqref{eq:m}, first notice that $H(\lambda, M')$ can be lower bounded  as
\begin{align*}
H(\lambda, M')&={M' \choose \lceil\lambda M'\rceil}^{-1}\\
&\geq 2^{- \lceil\lambda M'\rceil\log \frac{eM'}{ \lceil\lambda M'\rceil}}\\
&\geq 2^{-(\lambda M'+1)\log \frac{e}{\lambda}}
\end{align*}
where we used the upper bound for binomial coefficient ${n\choose k}\leq (ne/k)^k=2^{k\log(en/k)}$. Then we lower bound $\tilde N$ as
\begin{align}
\tilde N&= H(\lambda,M')M'^{-1}\left( \frac{1-\epsilon}{2\epsilon}\right)^{\lceil\lambda M'\rceil}\geq 2^{-(\lambda M'+1)\log \frac{e}{\lambda}}M'^{-1}\left( \frac{1-\epsilon}{2\epsilon}\right)^{\lceil\lambda M'\rceil}\nonumber\\
&\geq  M'^{-1} 2^{M'\lambda \left(\log((1-\epsilon)/2\epsilon)-\log \frac{e}{\lambda}\right) + \log\frac{e}{\lambda} }\nonumber\\
&\geq \frac{1}{\epsilon M+1}2^{M\epsilon\lambda \left(\log((1-\epsilon)/2\epsilon)-\log \frac{e}{\lambda}\right) }\label{eq:N_lb_m_beta}
\end{align}
where we used $\lceil\lambda M'\rceil\geq \lambda M'$, $M'\leq \epsilon M+1$ and $\log\frac{e}{\lambda}\geq 0$ (because $\lambda\in(0,1/2)$) in the last two inequalities, respectively. With our  choice of $\epsilon$, we can lower bound the term in the exponent
\begin{align}
\log\frac{1-\epsilon}{2\epsilon}- \log \frac{e}{\lambda}
&= \log \left(\frac{2e}{\lambda}-\frac{1}{2} \right)-\log \frac{e}{\lambda} \nonumber \\
&= \log (2-\frac{\lambda}{2e}) \nonumber\\
&\geq  \log(2-1/4e)\label{eq:lb_exponent}
\end{align}
as $\lambda<1/2$.  Using  (\ref{eq:lb_exponent}) in  (\ref{eq:N_lb_m_beta}), and recalling $M=2^{-n(C-\xi)}$,  we can further lower bound $\tilde N$ as (here we use $\exp\{\cdot\}$ to denote $2^{\{\cdot\}}$)
\begin{align*}
\tilde N&\geq \frac{1}{\epsilon M + 1}\exp\left\{\frac{\log(2-1/4e)}{4e}\cdot 2^{(C-\xi)n}\cdot 2^{-\frac{2\beta C}{1+2\beta}n}\right\}\\
&\geq \frac{1}{M+1} \exp\left\{C' 2^{(\frac{ C}{1+2\beta}-\xi)n}\right\}\\
&\geq \frac{1}{|\mc X|^n+1} \exp\left\{C' 2^{(\frac{ C}{1+2\beta}-\xi)n}\right\}
\end{align*}
%\begin{align*}
%N&\geq \frac{1}{\epsilon M + 1}2^{\frac{\log(2-1/4e)}{4e}\cdot 2^{(C-\xi)n}\cdot 2^{-\frac{2\beta C}{1+2\beta}n} - \frac{\beta C}{1+2\beta} n \log e}\\
%&\geq \frac{1}{M+1} 2^{C' 2^{(\frac{ C}{1+2\beta}-\xi)n}- \frac{\beta C}{1+2\beta} n \log e}
%\end{align*}
where we define $C':=\frac{\log(2-1/4e)}{4e}$.  Now to show \eqref{eq:m} in this case, it suffices to show
\begin{align*}
2^m\leq  \frac{1}{|\mc X|^n+1} \exp\left\{C' 2^{(\frac{ C}{1+2\beta}-\xi)n}\right\}
\end{align*}
Substituting the choice $m=\left\lceil2^{(\frac{C}{1+2\beta}-\xi')n}\right\rceil$ in the above expression and taking the logarithm on both sides, we have
\begin{align*}
\left\lceil2^{(\frac{C}{1+2\beta}-\xi')n}\right\rceil\leq C' 2^{(\frac{ C}{1+2\beta}-\xi)n}-\log (|\mc X|^n+1)
\end{align*}
The above inequality is satisfied if
\begin{align*}
2^{(\frac{C}{1+2\beta}-\xi')n+1}\leq C' 2^{(\frac{ C}{1+2\beta}-\xi)n}-\log (|\mc X|^n+1)
\end{align*}
or equivalently
\begin{align*}
2^{-(\xi'-\xi)n+1}\leq C' -2^{-(\frac{ C}{1+2\beta}-\xi)n}\log (|\mc X|^n+1)
\end{align*}
Observe that as $n\rightarrow \infty$, the RHS of the above inequality approaches $C'>0$ and the LHS approaches $0$ for  $\xi'>\xi$. Hence we have verified \eqref{eq:m}.

Now we upper bound the error probability. Using $S=cm^{\beta}$, the false positive error probability is upper bounded as
\begin{align*}
\lambda_2&= S\lambda+\delta\\
&=cm^\beta\lambda+\delta\\
&\stackrel{(a)}{=} c  (\lceil2^{(\frac{C}{1+2\beta}-\xi')n}\rceil)^{\beta} 2^{-\frac{\beta C}{1+2\beta}n}+\delta\\
&\leq c  2^{(\frac{C}{1+2\beta}-\xi')\beta n+\beta} 2^{-\frac{\beta C}{1+2\beta}n}+\delta\\
&=c 2^{-\beta\xi' n+\beta}+\delta
\end{align*}
where in $(a)$ we substitute our choice of $m$ and $\lambda$. Clearly, the error probability can be made arbitrarily small for large enough $n$ and a sufficiently small $\delta$. Therefore we have shown that the error probability can be made arbitrarily small with $m>2^{(\frac{C}{1+2\beta}-\xi')n}$ for $\xi'>\xi>0$. In other words, $C/(1+2\beta)$ is an achievable computation rate with the scaling function $L(R,n)=2^{Rn}$.

\textbf{Case 3 ($S=cm^{\log m}$).} In this  case we choose $\lambda = 2^{-(C-\xi)n/2+\sqrt{Cn/2}}$, $\epsilon = 2^{-(C-\xi)n/2}$ and  $m=\left\lceil 2^{\sqrt{(C/2-\xi')n}}\right\rceil$ for some $\xi'$ chosen as in Case 1. The condition  \eqref{eq:epsilon_const} is satisfied for large enough $n$. To verify \eqref{eq:m},  we use the bound  derived in \eqref{eq:N_lb_m_beta}, which we restate here
\begin{align*}
\tilde N &\geq \frac{1}{2\epsilon M}2^{M\epsilon\lambda (\log((1-\epsilon)\lambda/\epsilon)-\log 2e) }
\end{align*}
We can lower bound the term
\begin{align*}
\log \frac{(1-\epsilon)\lambda}{\epsilon} &= \log (1- \epsilon)  +\sqrt{Cn/2}\\
&\geq \sqrt{Cn/2}-1
\end{align*}
where the last step holds as $\epsilon<1/6<1/2$.   Recall that $M=2^{(C-\xi)n}$. With the choice of $\lambda$ and $\epsilon$ we have $\lambda\epsilon M = 2^{\sqrt{Cn/2}}$ and 
\begin{align*}
\tilde N&\geq \frac{1}{2}\exp\left\{-(C-\xi)n/2+2^{\sqrt{Cn/2}}(\sqrt{Cn/2}-\log 4e)\right\}
\end{align*}
With the choice $m=\left\lceil 2^{\sqrt{(C/2-\xi')n}}\right\rceil$ and the inequality above, \eqref{eq:m} holds if
\begin{align*}
\left\lceil 2^{\sqrt{(C/2-\xi')n}}\right\rceil\leq -1 -(C-\xi)n/2 +2^{\sqrt{Cn/2}}(\sqrt{Cn/2}-\log 4e)
\end{align*}
which holds if
\begin{align*}
2^{\sqrt{(C/2-\xi')n}+1}\leq -1 -(C-\xi)n/2 +2^{\sqrt{Cn/2}}(\sqrt{Cn/2}-\log 4e)
\end{align*}
Dividing both sides by $2^{\sqrt{Cn/2}}$, the above inequality becomes
\begin{align*}
2^{\sqrt{(C/2-\xi')n}-\sqrt{Cn/2}+1}\leq \sqrt{Cn/2}-\log 4e + 2^{-\sqrt{Cn/2}}(-1 -(C-\xi)n/2)
\end{align*}
The LHS can be written as $2^{\frac{-\xi' n}{\sqrt{(C/2-\xi')n}+\sqrt{Cn/2}}+1}$ (since $\sqrt{a}-\sqrt{b}=(a-b)/(\sqrt{a}+\sqrt{b})$) which converges to $0$ as $n\rightarrow \infty$  whereas the RHS scales as $\sqrt{Cn/2}$ which tends to infinity as $n$ increases. This verified that  \eqref{eq:m} holds for large enough $n$.

The false positive error  is upper bounded as
\begin{align*}
\lambda_2&= S\lambda+\delta\\
&=c2^{(\log m)^2}\lambda+\delta\\
&\leq c2^{\left(1+\sqrt{(C/2-\xi')n}\right)^2}2^{-(C-\xi)n/2+\sqrt{Cn/2}}+\delta\\
&=c2^{-(\xi'-\xi)n+2\sqrt{(C/2-\xi')n}+\sqrt{Cn/2}+1}+\delta.
\end{align*}
Given our choice $\xi'>\xi$, the above quantity  can be made arbitrarily small for a sufficiently large $n$ and a sufficiently small $\delta$. Thus we have shown that $m\geq 2^{\sqrt{(C/2-\xi')n}}$ for all $\xi'>0$. In other words, $C/2$ is an achievable computation rate with the scaling function $L(R,n)=2^{\sqrt{Rn}}$.

\textbf{Case 4 $S=c2^{m^{1/b}}$ for $b>1$.} In this case, we first impose an extra assumption $C>3$, and write $C=3+\delta'$ for some $\delta'>0$. We choose $\lambda = n^{b-1}2^{-(C/3-\xi/2)n}$, $\epsilon = 2^{-(2C/3-\xi/2)n}$, and $m = \lceil \tilde d n^{b}\rceil$ with $\tilde d=C/3-\xi'$ for some $\xi'\in (\xi/2, \delta'/3]$. Notice that this choice of $\epsilon$ satisfies the condition \eqref{eq:epsilon_const} for large $n$. 

To verify \eqref{eq:m}, we use the bound  derived in \eqref{eq:N_lb_m_beta}, which we restate here
\begin{align*}
\tilde N &\geq \frac{1}{2\epsilon M}2^{M\epsilon\lambda (\log((1-\epsilon)\lambda/\epsilon)-\log 2e)  }
\end{align*}
We can lower bound the term
\begin{align*}
\log \frac{(1-\epsilon)\lambda}{\epsilon} &= \log (1- \epsilon) + \log 2^{Cn/3} +(b-1)\log n\\
&\geq \frac{C}{3} n - 1 + (b-1)\log n
\end{align*}
where the last step holds as $\epsilon<1/6<1/2$.   Using  $M=2^{(C-\xi)n}$ and our choice of $\lambda$ and $\epsilon$, we have $\lambda\epsilon M = n^{b-1}$ and (use $\exp\{\cdot\}$ to denote $2^{\{\cdot\}}$)
\begin{align*}
\tilde N&\geq 2^{-1}2^{-(C/3-\xi/2)n} \exp\{n^{b-1}(Cn/3-1+(b-1)\log n - \log 2e)  \}\\
&=2^{-1}\exp\{Cn^{b}/3+ n^{b-1}((b-1)\log n -\log 4e)-n(C/3-\xi/2)\}
\end{align*}
In light of the inequality above, \eqref{eq:m} holds if
\begin{align*}
\lceil \tilde d n^{b}\rceil\leq -1+ \frac{C}{3} n^{b}+ n^{b-1}((b-1)\log n -\log 4e)-n(C/3-\xi/2)
\end{align*}
which holds if
\begin{align*}
\tilde d n^{b}\leq -2+ \frac{C}{3} n^{b}+ n^{b-1}((b-1)\log n -\log 4e)-n(C/3-\xi/2)
%\label{eq:m_1/b}
\end{align*}
%TODO b=1 case, not an optimal construction
%We first consider the special case when $b=1$. In this case the above inequality becomes
%\begin{align*}
%2\tilde dn\leq n(4d-C+\xi)-2\log 4e
%\end{align*}
%Choosing $\tilde d = d = \frac{C-\xi'}{2}$ for some $0<\xi'<\xi$, the above inequality is equivalent to
%\begin{align*}
%0\leq \xi-\xi'-n^{-1}\log 4e
%\end{align*}
%which is valid for sufficiently large $n$.  This shows that in this case we can choose $m = \lceil  (C-\xi')n\rceil$.
Dividing both sides of the above inequality by $\tilde d n^{b}$ gives
\begin{align*}
1\leq \frac{C/3}{\tilde d} + \frac{(b-1)\log n -\log 4e}{\tilde d}n^{-1}-\frac{2}{\tilde d}n^{-b} - \frac{C/3-\xi/2}{\tilde d}n^{-(b-1)}
\end{align*}
Notice that for $b>1$, the RHS approaches $\frac{C/3}{\tilde d}$ as $n\rightarrow \infty$. As we choose $\tilde d= C/3-\xi'$, the inequality holds for sufficiently large $n$.

The false positive error probability is
\begin{align*}
\lambda_2&= S\lambda+\delta\\
&=c2^{m^{1/b}}\lambda+\delta\\
&\leq c2^{(\tilde d n^b+1)^{1/b}}2^{-(C/3-\xi/2)n}n^{b-1}+\delta&\\
&\stackrel{(a)}{\leq} c2^{\tilde d^{1/b} n+1}2^{-(C/3-\xi/2)n}n^{b-1}+\delta&\\
&=c2^{-n(C/3-\tilde d^{1/b}-\xi/2)+1}n^{b-1}+\delta
\end{align*}
where in (a) we use $(x+y)^a\leq x^a+y^a$ for $a\in(0,1], x,y\geq 0$. Notice due to the assumption $\xi'\leq \delta'/3$, we have $\tilde d = C/3-\xi'\geq (C-\delta')/3=1$, where the last equality follows from $C=3+\delta'$. Therefore it holds that $(\tilde d)^{1/b}\leq \tilde d$ since $\tilde d\geq 1$.  The error probability can be further upper bounded as
\begin{align*}
\lambda_2&\leq c2^{-n(C/3-\tilde d-\xi/2)+1}n^{b-1}+\delta\\
&= c2^{-n(C/3-C/3+\xi'-\xi/2)+1}n^{b-1}+\delta\\
&=c2^{-n(\xi'-\xi/2) +1}n^{b-1}+\delta
\end{align*}
which tends to $0$ as $n\rightarrow\infty$ under the assumption $\xi'>\xi/2$. Therefore we have shown that the error probability vanishes with  $m\geq  \tilde d n^b=(C/3-\xi')n^{b}$ for all  $\xi'>\xi/2$. In other words, $C/3$ is an achievable computation rate for the scaling function $L(R,n)=Rn^b$.

Now consider the case when the capacity of the channel $C$ lies in the interval $(0,3]$. The strategy here is to form a ``super-channel" $W'$ with multiple original channels so that $W'$ has capacity larger than $3$, and our result applies. Specifically, we take the ``super-channel" $W'$ to be the cartesian product of $T:=\lfloor \frac{3}{C}+1\rfloor$ original channel $W$ with capacity $C>0$ whose input and output lie in the space $\mc X^{\otimes T}$ and $\mc Y^{\otimes T}$, respectively. Then it is clear that the capacity $C'$ of $W'$ is given by $C'=CT>C(3/C-1+1)=3$.

Now we need to characterize the effective achievable computation rate by converting our original channel (with $n$ channel uses) to the new channel. Assume that we use the new channel (with a capacity $C'>3$) $\tilde n$ times, the results derived above show that $m\geq (C'/3-\xi')\tilde n^b$. On the other hand, we have the relationship $\tilde n= n/T$ (because every $T$ original channels are used to form a new channel) and $C'=CT$. Plugging the two expressions into the inequality above, we come to the conclusion that $\frac{C}{3T^{b-1}}$ is achievable for the scaling function $L(R,n)=Rn^b$, where $T:=\lfloor \frac{3}{C}+1\rfloor$.

{\bf Case 5 ($S=c2^{m/\log m}$)}. We choose $\lambda=n\log n 2^{-(C-\xi)n/2}$, $\epsilon= 2^{-(C-\xi)n/2}$ and $m=\left\lceil \frac{(C-\xi')}{2}n\log n\right\rceil$ for some $\xi'>\xi$. This choice of $\epsilon$ satisfies the condition \eqref{eq:epsilon_const} for large $n$. To verify \eqref{eq:m}, we use the bound  derived in \eqref{eq:N_lb_m_beta}, which implies
\begin{align*}
\tilde N\geq \frac{1}{2\epsilon M}2^{M\epsilon\lambda (\log((1-\epsilon)\lambda/\epsilon)-\log 2e)  }
\end{align*}
We can lower bound the term
\begin{align*}
\log \frac{(1-\epsilon)\lambda}{\epsilon}&=\log (1-\epsilon)+\log n+\log\log n\\
&\geq \log n+\log\log n-1
\end{align*}
as $\epsilon<1/2$. Recall that $M=2^{(C-\xi)n}$ hence $\lambda\epsilon M=n\log n$, which implies
\begin{align*}
\tilde N&\geq 2^{-1}2^{(C-\xi)n/2}2^{-(C-\xi)n}2^{n\log n (\log n+\log\log n-1-\log 2e)}\\
&=2^{-1-(C-\xi)n+ n\log n(\log n+\log\log n-\log 4e)}
\end{align*}

With the choice of $m$ in this case, \eqref{eq:m} holds if
\begin{align*}
\left\lceil \frac{(C-\xi')}{2}n\log n\right\rceil \leq -1-(C-\xi)n+ n\log n(\log n+\log\log n-\log 4e)
\end{align*}
which holds if
\begin{align*}
 \frac{(C-\xi')}{2}n\log n \leq -2-(C-\xi)n+ n\log n(\log n+\log\log n-\log 4e)
\end{align*}
Dividing both sides by $n\log n$ we have
\begin{align*}
(C-\xi')/2\leq \log n+\log\log n-\log 4e -\frac{2+
(C-\xi)n}{n\log n}
\end{align*}
Clearly, the above inequality holds for large enough $n$ as the RHS grows as $O(\log n)$ whereas the LHS is a constant.

The false positive error probability is upper bounded as
\begin{align*}
\lambda_2&=S\lambda+\delta\\
&=c2^{m/\log m}\lambda+\delta\\
&\leq c\exp\left\{\frac{\frac{C-\xi'}{2} n\log n+1}{\log (C-\xi')/2+\log n+\log \log n}\right\}\lambda+\delta\\
&\leq c\exp\left\{\frac{\frac{C-\xi'}{2} n\log n+1}{\log n}\right\} 2^{-(C-\xi)n/2}n\log n +\delta\\
&=c n\log n2^{-(\xi'-\xi)n/2+1/\log n}+\delta
\end{align*}
For all $\xi'>\xi$, the above quantity can be made arbitrarily small for sufficiently large $n$ and a sufficiently small $\delta$. We have shown that $m\geq \frac{C-\xi'}{2} n\log n$ for  $\xi'>\xi>0$. Since $\xi$ can also be made arbitrarily small, we have shown that $C/2$ is an achievable computation rate with the scaling function $L(R,n)=Rn\log n$.

{\bf Case 6 ($S=c2^{\gamma m})$ for $\gamma\in(0,1]$}. In this case we claim that $C$ is an achievable computation rate with the scaling function $L(R,n)=Rn$. This rate is in fact the same as the achievable transmission rate from Shannon's problem. In other words, to achieve this computation rate, we could use the original error correction codes (described in \eqref{eq:good_channel_code}) and use a ``separation approach", namely to let the receiver recover the message $i$ reliably, and then compute $f_j(i)$ at the receiver. In this sense, the encoder/decoder construction at the beginning of this section is not needed for this result. Nevertheless, we could also proceed in the same way as in the other cases to derive this result. This analysis, though not necessary, is included in Appendix for completeness.

%{\bf Case  ($S=c 2^m$):} We consider the case when $S=c2^m$ for some constant $c$ not depending on $m,n$.  In this case we choose $M'=1$ and Proposition  \ref{prop:maximal_code} states that  in this case we have $N=M$.  Using the rate function $\log(\cdot)$, we have
%\begin{align*}
%\frac{1}{n}\log 2^m &\geq \frac{1}{n}\log N/2\\
%&=\log M/2\geq  \frac{1}{n}\log (2^{n(C-\xi)-1})\\
%&=C-\xi-1/n
%\end{align*}
%We choose $\lambda=2^{-nC}$. In this case,  the false positive error is upper bounded as
%\begin{align*}
%Q_iW(\mc D_j)&\leq S\lambda+\delta\\
%&=c2^m \lambda+\delta\\
%&\leq cN\lambda+\delta\\
%&=c (2^{n(C-\xi)})2^{-nC}+\delta\\
%&=c 2^{-n\xi}+\delta
%\end{align*}
%which can be made arbitrarily small for large enough $n$ and by choosing $\delta$ small enough.

\section{Proof of the converse}
\label{sec:proof_converse}

For the converse results, we follow the idea briefly discussed in \cite{ahlswede_general_2008}. The strategy is to show that any $(n,m,\mc F, \lambda_1, \lambda_2)$ BFC code can be converted to a standard identification code, and the known converse result on identification code will imply a converse result for the BFC code. Recall the definition of the $(n,N,\lambda_1,\lambda_2)$ ID code in Section \ref{sec:Ahlswede_models}. Define $N^*(n,\lambda_1,\lambda_2)$ to be the maximal $N$ such that an $(n,N, \lambda_1,\lambda_2)$ ID code exists. We have the following known converse result for ID codes.

\begin{theorem}[\cite{ahlswede_identification_1989,ahlswede_strong_2002}]
Let $\lambda_1,\lambda_2>0$ such that $\lambda_1+\lambda_2<1$. Then for every $\delta>0$ and every sufficiently large $n$
\begin{align*}
N^*(n, \lambda_1,\lambda_2)\leq 2^{2^{n(C+\delta)}}
\end{align*}
where $C$ is the Shannon capacity of the channel.
\label{thm:converse_id}
\end{theorem}

To prove the result in Theorem \ref{thm:converse}, we in fact show that there is a subset $\mc F\subset \wcf(S)$ such that any BFC code for the set $\mc F$ satisfies the statement in Theorem \ref{thm:converse}. Specifically, we prove the following result.
%\begin{theorem}[Converse for a subset]
%Assume the channel $W(\cdot | x^n)$ has a Shannon capacity $C>0$.  Let $\lambda_1,\lambda_2>0$ such that $\lambda_1+\lambda_2<1$, and let $m^*$ be the maximal $m$ such that such an $(n,m, \mc F,\lambda_1,\lambda_2)$ BFC code exists. Then there exists a set of constant weight functions $\mc F\subset \wcf(S)$, such that the corresponding $m^*$ for the $(n,m,  \mc F, \lambda_1,\lambda_2)$ BFC code satisfy the statements in Theorem \ref{thm:converse}.
%\label{thm:converse_subset}
%\end{theorem}

\begin{theorem}[Converse for a subset]
Assume the channel $W(\cdot | x^n)$ has a Shannon capacity $C>0$. Then there exists a set of constant weight functions $\mc F\subset \wcf(S)$, such that the computation capacity $\cbfc$ for $\mc F$ satisfies the statements in Theorem \ref{thm:converse}.
\label{thm:converse_subset}
\end{theorem}
%Assume the channel $W(\cdot | x^n)$ has a Shannon capacity $C>0$.  Let  $\cbfc$ denote the computation capacity for $\wcf(S)$ with a  scaling function specified below. 

Since  $\mc F$ is contained in $\wcf(S)$, any converse result on the code $(n,m,\mc F,\lambda_1,\lambda_2)$ automatically implies a converse result on the code  $(n,m,\wcf(S),\lambda_1,\lambda_2)$, which proves Theorem \ref{thm:converse}. The proof of Theorem \ref{thm:converse_subset} relies on the following Lemma \ref{lemma:bcf_to_id}, which is proved using Lemma \ref{lemma:gilbert_cw} (Gilbert's bound for constant weight sequences) in Appendix. Lemma \ref{lemma:bcf_to_id} states that we can find a subset $\mc F\subset\wcf(S)$ which has essentially the same size as $\wcf(S)$, but whose elements have small overlaps in their pre-images. The latter property allows us to convert a BFC code into a good identification code.

\begin{lemma}
For any integer $m\geq 1$, $1\leq S\leq 2^{m-1}$ and $\alpha\in(0,1]$, there exists a set of constant weight functions $\mc F\subset \wcf(S)$ satisfying the following properties
\begin{itemize}
\item for any $f_i, f_j\in \mc F$, it holds that $|f_\ell^{-1}[1]\cap f_j^{-1}[1]|\leq \alpha S$
\item $\log |\mc F|\geq \log {2^m\choose S}-\log {2^m\choose \lceil (1-\alpha)S\rceil-1 }-S$
%\item $\log |\mc F|\geq S(\alpha m -\log S-1)$
\end{itemize}
\label{lemma:bcf_to_id}
\end{lemma}
\begin{proof}
Any function $f\in\mc F_m$ can be represented by  a  binary sequence $\ve f$ of length $2^m$, where the $i$-th position of $\ve f$ is $1$ if and only if $f(i')=1$. Here $ i'$ is the dyadic representation of $i$, i.e., the unique binary sequence of length $m$ such that $i=\sum_{j=0}^{m-1}2^j i'_{j+1}$.  In other words, $\ve f$ is defined via $\ve f_i:=f(i')$. It is clear that if $f$ has a Hamming weight $S$, the corresponding vector $\ve f$ also has Hamming weight $S$. So we can apply Lemma \ref{lemma:gilbert_cw} by identifying $T=2^m$ and $G=\mc F$. 
%To show the bound we notice that
%\begin{align*}
%\log |\mc F|&\geq \log {2^m\choose S}-\log {2^m\choose \lceil (1-\alpha)S\rceil-1 }-S\\
%&\stackrel{(a)}{\geq} S\log \frac{2^m}{S} - ( \lceil (1-\alpha)S\rceil-1 )\log 2^m-S\\
%&\geq  S\log \frac{2^m}{S} - (1-\alpha)S\log 2^m-S\\
%&=S(\alpha m -\log S-1)
%\end{align*}
%where in (a) we used $(n/k)^k\leq {n \choose k}\leq n^k$.
\end{proof}

%\begin{lemma}
%There exists a set of constant weight functions $\mc F\subset \wcf(S)$ such that 
%\end{lemma}

\begin{proof}[Proof of Theorem \ref{thm:converse_subset}]
We consider a set of functions $\mc F$ with properties specified  by Lemma \ref{lemma:bcf_to_id}, and consider  an  $(n,m, \mc F,\lambda_1,\lambda_2)$ BFC code characterized by $(Q_i)_{i\in\{0,1\}^m}, (D_j)_{j\in \intset{|\mc F|}}$  for this class of functions.

We first show that we can convert the above BFC code to an identification code as follows. Specifically, the new identification code is defined to be the encoder/decoder pairs $\{(\tilde Q_j, D_j), j\in\intset{|\mc F|}\}$ where encoder $\tilde Q_j$ is defined as
\begin{align*}
\tilde Q_j(x^n):=\frac{1}{S}\sum_{i\in f_j^{-1}[1]}Q_i(x^n)
\end{align*}
Now we characterize the error probability of this constructed identification code. The mis-identification error is 
\begin{align*}
\sum_{\ve x} \tilde Q_j(\ve x)W(D_j^c|\ve x)&=\sum_{\ve x}\frac{1}{S}\sum_{i\in f_j^{-1}[1]}Q_i(\ve x)W(D_j^c|\ve x)\\
&=\frac{1}{S}\sum_{i\in f_j^{-1}[1]}\sum_{\ve x}Q_i(\ve x)W(D_j^c|\ve x)\\
&\leq \frac{1}{S}\sum_{i\in f_j^{-1}[1]} \lambda_1\\
&=\lambda_1
\end{align*}
where the inequality holds because $f_j(i)=1$ hence $\sum_{\ve x}Q_i(\ve x)W(D_j^c|\ve x)\leq\lambda_1$, by the property of the BFC code. The last equality holds  because $|f_{j}^{-1}[1]|=S$ as $f_j$ has Hamming weight $S$.

The wrong-identification error is, for $\ell\neq j$, 
\begin{align*}
\sum_{\ve x} \tilde Q_\ell(\ve x)W(D_j|\ve x)&=\sum_{\ve x}\frac{1}{S}\sum_{i\in f_\ell^{-1}[1]}Q_i(\ve x)W(D_j|\ve x)\\
&=\frac{1}{S}\left(\sum_{i\in f_\ell^{-1}[1]\cap f_j^{-1}[1]}\sum_{\ve x}Q_i(\ve x)W(D_j|\ve x)+\sum_{i\in f_\ell^{-1}[1]\backslash f_j^{-1}[1]}\sum_{\ve x}Q_i(\ve x)W(D_j|\ve x)\right)\\
&\stackrel{(a)}{\leq} \frac{1}{S}\left(|f_\ell^{-1}[1]\cap f_j^{-1}[1]|+  \sum_{i\in f_\ell^{-1}[1]\backslash f_j^{-1}[1]}\lambda_2\right)\\
&\stackrel{(b)}{\leq} \frac{|f_\ell^{-1}[1]\cap f_j^{-1}[1]|}{S}+\lambda_2\\
&\stackrel{(c)}{\leq} \alpha+\lambda_2
\end{align*}
where in (a) we use the upper bound $\sum_{\ve x}Q_i(\ve x)W(D_j|\ve x)\leq 1$ and the BFC code error property, step (b) follows $|f_\ell^{-1}[1]\backslash f_j^{-1}[1]|\leq S$, and step (c) follows from the first property stated in Lemma \ref{lemma:bcf_to_id}. 

Therefore, we have shown that  an $(n,m, \mc F, \lambda_1,\lambda_2)$ BFC codes (with properties of $\mc F$ prescribed in Lemma \ref{lemma:bcf_to_id}) implies the existence of an $(n, |\mc F|, \lambda_1, \alpha+\lambda_2)$ ID code for all $\alpha\in(0,1)$. Now fix some $\alpha\in(0,1-\lambda_1-\lambda_2)$.  Theorem \ref{thm:converse_id} states that for any $\lambda_1,\lambda_2$ such that $\lambda_1+\lambda_2<1$,  it holds that
\begin{align}
\log |\mc F|\leq 2^{n(C+\delta)}
\label{eq:converse_F}
\end{align}
for every $\delta>0$ and every sufficiently large $n$. %For the argument below, we consider an arbitrary but  fixed  $\delta$.

Now we lower bound $|\mc F|$ in terms of $m$ using the property of $\mc F$ shown in Lemma \ref{lemma:bcf_to_id} for different cases of $S$, which will, in turn, give an upper bound on $m$.

{\bf Case 1 (constant $S$).} To apply the result in Lemma \ref{lemma:bcf_to_id}, we need $S\leq 2^{m-1}$, which is clearly satisfied for $m$ large enough. In this case we have
\begin{align}
\log |\mc F|&\geq \log {2^m\choose S}-\log {2^m\choose \lceil (1-\alpha)S\rceil-1 }-S \nonumber\\
&\stackrel{(a)}{\geq} S\log \frac{2^m}{S} - ( \lceil (1-\alpha)S\rceil-1 )\log 2^m-S \nonumber\\
&\geq  S\log \frac{2^m}{S} - (1-\alpha)S\log 2^m-S\nonumber\\
&=S(\alpha m -\log S-1) \label{eq:lowerbound_F}\\
&=m\underbrace{(\alpha S-m^{-1}S\log S- S/m)}_{A(m)} \nonumber
\end{align}
where in (a) we used $(n/k)^k\leq {n \choose k}\leq n^k$. Combining the above inequality with \eqref{eq:converse_F}, we know that the following inequality  holds
\begin{align*}
mA(m)\leq 2^{n(C+\delta)}
\end{align*}
for every $\delta>0$ and every sufficiently large $n$. Notice $A(m)$ approaches $\alpha S>0$  from the left as $m\rightarrow \infty$. So we can choose $m$ large enough such that $A(m)\geq \alpha S/2$. In this case we have 
\begin{align*}
m\leq\frac{2}{\alpha S}2^{n(C+\delta)}\leq 2^{n(C+\delta')}
\end{align*}
for all $\delta'>\delta>0$ and  sufficiently large $n$. Since the above inequality holds for all $\delta'>0$, it follows that $m\leq 2^{nC}$ hence $C$ is an upper bound on the computation capacity with the scaling function $L(R,n)=2^{Rn}$.

{\bf Case 2 ($S=c m^\beta$)}  for $\beta \geq 0$. To apply the result in Lemma \ref{lemma:bcf_to_id}, we need $cm^{\beta}\leq 2^{m-1}$, which is clearly satisfied for $m$ large enough. We reuse the lower bound \eqref{eq:lowerbound_F} derived in Case 1 for a general $S$
\begin{align*}
\log |\mc F|&\geq S(\alpha m -\log S-1)\\
&=cm^{\beta}(\alpha m - \log c - \beta\log m-1)
\end{align*}
Combining the above inequality with \eqref{eq:converse_F}, we know that the following inequality must hold
\begin{align*}
cm^{\beta}(\alpha m - \log c - \beta\log m-1)\leq 2^{n(C+\delta)}
%\alpha cm^{\beta+1}-\beta cm^{\beta}\log cm - cm^\beta\leq 2^{n(C+\delta)}
\end{align*}
for every sufficiently large $n$. This is equivalent to
\begin{align*}
m^{1+\beta}(\alpha c-c m^{-1}\log c-c\beta m^{-1}\log m- cm^{-1})\leq 2^{n(C+\delta)}
\end{align*}
We choose $m$ large enough such that $c m^{-1}\log c+c\beta m^{-1}\log m+ cm^{-1}\leq \alpha c/2$ so the above inequality implies
\begin{align*}
m^{1+\beta}\leq \frac{2}{\alpha c} 2^{n(C+\delta)}\leq 2^{n(C+\delta')}
\end{align*}
for all $\delta'>\delta$ and sufficiently large $n$ and $m$. Since it holds  for all $\delta'>0$, we have shown that 
\begin{align*}
m\leq 2^{ \frac{C}{1+\beta}n}
\end{align*}
for sufficiently large $m$ and $n$. We showed that $C/(1+\beta)$ is an upper bound on the computation capacity with the scaling function $L(R,n)=2^{Rn}$.

\textbf{Case 3: $S=m^{\log m}$}. The condition in Lemma \ref{lemma:bcf_to_id} ($m^{\log m}=2^{(\log m)^2}\leq 2^{m-1}$) is satisfied for large $m$. Following the same step until \eqref{eq:lowerbound_F} in Case 1, we have
\begin{align*}
\log |\mc F|&\geq S(\alpha m -\log S-1)\\
&=2^{(\log m)^2}(\alpha m - (\log m)^2-1)
\end{align*}
Combining the above inequality with \eqref{eq:converse_F}, we know that the following inequality must hold
\begin{align*}
2^{(\log m)^2}(\alpha m - (\log m)^2-1)\leq 2^{n(C+\delta)}
\end{align*}
For $m$ large enough, we have
\begin{align*}
2^{(\log m)^2}\leq \frac{1}{\alpha m - (\log m)^2-1}2^{n(C+\delta)}\leq 2^{n(C+\delta')}
\end{align*}
for all $\delta'>\delta>0$ and sufficiently large $n$, which is equivalent to $m\leq 2^{\sqrt{Cn}}$. Therefore we have shown that $\sqrt{C}$ is an upper bound on the computation capacity with the scaling function $L(R,n)=2^{R\sqrt{n}}$.

\textbf{Case 4}: $S=c2^{m^{1/b}}$ for $b>1$. The assumption in Lemma \ref{lemma:bcf_to_id} ($S\leq 2^{m-1}$) is satisfied for all $b>1$ and $m$ large enough. Let $d:= \lceil (1-\alpha)S\rceil-1$, then we have the following lower bound by Lemma \ref{lemma:bcf_to_id}
\begin{align}
\log |\mc F|&\geq \log {2^m\choose S}-\log {2^m\choose d }-S \nonumber\\
&\stackrel{(a)}{\geq} S\log \frac{2^m}{S} -  d\log (2^me/d)-S\nonumber\\
&\geq  S\log \frac{2^m}{S} - (1-\alpha)S(\log e2^m-\log d)-S\nonumber\\
&\geq  S\log \frac{2^m}{S} - (1-\alpha)S(\log e2^m-\log ((1-\alpha)S-1))-S \label{eq:F_lowerbound_2}\\
% &= S(m-\log S-(1-\alpha)\log e -(1-\alpha)m+(1-\alpha)\log\tilde d-1)\\
% &=c2^{m^{1/b}}(\alpha m-\log c -m^{1/b} -(1-\alpha)\log\tilde d-1-(1-\alpha)\log e)\\ 
&=m2^{m^{1/b}}(c\alpha - cm^{-1}\log c-cm^{1/b-1}-cm^{-1}-cm^{-1}(1-\alpha)\log e-cm^{-1}(1-\alpha)\log ((1-\alpha)S-1))\nonumber\\
&\stackrel{(b)}{=}m2^{m^{1/b}}\underbrace{\left(c\alpha+c\alpha m^{1/b-1}+cm^{-1}(1-\alpha\log (1-\alpha-S^{-1}))-cB/m \right)}_{A(m)} \nonumber
\end{align}
where in (a) we used $(n/k)^k\leq {n \choose k}\leq (en/k)^k$,  and in (b) we define $B:=\log c+(1-\alpha)\log(e/c)+1$. Combining the above inequality with \eqref{eq:converse_F}, we have
\begin{align*}
mA(m)2^{m^{1/b}}\leq 2^{n(C+\delta)}
\end{align*}
where  $A(m)\rightarrow c\alpha$ as $m\rightarrow \infty$. This implies
\begin{align*}
m2^{m^{1/b}}\leq 2^{n(C+\delta')}
\end{align*}
for all $\delta'>\delta>0$ and sufficiently large $n$ and $m$, which is equivalent to $m2^{m^{1/b}}\leq 2^{nC}$. On the other hand we have $m2^{m^{1/b}}\geq 2^{m^{1/b}}$. This shows $m$ must satisfy
\begin{align*}
m\leq (Cn)^b
\end{align*}
for sufficiently large $m, n$.

\textbf{Case 5:} $S=c2^{m/\log m}$.   To apply the result in Lemma \ref{lemma:bcf_to_id}, we need $c2^{m/\log m}\leq \frac{1}{2} 2^{m}$, which is clearly satisfied  for large enough $m$. Following the same step until \eqref{eq:F_lowerbound_2} as in Case (4), we have 
\begin{align*}
\log |\mc F|&\geq   S\log \frac{2^m}{S} - (1-\alpha)S(\log e2^m-\log ((1-\alpha)S-1))-S\\
&= S\left(\alpha(m-\log S)+(1-\alpha)\log(1-\alpha-1/S) -(1-\alpha)\log e-1 \right) \\
&\stackrel{(a)}{=}m2^{m/\log m}\underbrace{\left(\alpha c\left(1-\frac{\log c}{m}-\frac{1}{\log m}\right)+\frac{c(1-\alpha)}{m}\log (1-\alpha-1/c\cdot 2^{-m/\log m})-\frac{c(1-\alpha)\log e+c}{m}\right)}_{A(m)}
%&= m2^{\gamma m}\underbrace{\left(c\alpha(1-\gamma) +cm^{-1}\log\left(1-\frac{1}{(1-\alpha)c}2^{-\gamma m} \right)-cB/m\right)}_{A(m)}
\end{align*}
where in (a) we substitute $S=c2^{m/\log m}$. Combining the above inequality with \eqref{eq:converse_F}, we have
\begin{align*}
mA(m)2^{m/\log m}\leq 2^{n(C+\delta)}
\end{align*}
where $A(m)\rightarrow \alpha c$ as $m\rightarrow \infty$. This implies
\begin{align*}
m2^{m/\log m}\leq 2^{n(C+\delta')}
\end{align*}
for $\delta'>\delta$ and for all sufficiently large $m$ and $n$. As $m2^{m/\log m}\geq 2^{m/\log m}$, we have 
\begin{align}
\frac{m}{\log m}\leq n(C+\delta').
\label{eq:m_upperbound}
\end{align}
Given the above inequality, we show by contradiction that for all $\delta''>\delta'>0$
\begin{align}
m\leq (C+\delta'')n\log n
\label{eq:m_upperbound_final}
\end{align}
for all large  $m, n$. Assume the above is not true, namely, $\forall m_0, n_0$, there exist some  $m\geq m_0, n\geq n_0$ such that $m=r(C+\delta'')n\log n$ for some $r>1$. In this case, due to \eqref{eq:m_upperbound}, we have
\begin{align*}
\frac{m}{\log m}=\frac{r (C+\delta'')n\log n}{\log r (C+\delta'')+\log n+\log\log n}\leq n(C+\delta')
\end{align*}
which is equivalent to
\begin{align*}
r\leq \frac{C+\delta'}{C+\delta''}+\frac{(C+\delta')\log r(C+\delta'')+(C+\delta')\log \log n}{(C+\delta'')\log n}
\end{align*}
As the second term in the above inequality can be made arbitrarily small for with a large enough $n$ (by choosing $n_0$ large enough),  the above inequality leads to a contradiction as the RHS can be made smaller than $1$ while the LHS $r$ is assumed to be larger than $1$. This proves the claim in \eqref{eq:m_upperbound_final}. As \eqref{eq:m_upperbound_final} holds for all $\delta''>0$, we have prove that $m\leq C n\log n$ for all sufficiently $m,n$.

\textbf{Case 6:} $S= c 2^{\gamma m}$ for $\gamma\in(0,1)$.  To apply the result in Lemma \ref{lemma:bcf_to_id}, we need $c2^{\gamma m}\leq \frac{1}{2} 2^{m}$, which can be satisfied for any $\gamma<1$ for large enough $m$. Following the same step until \eqref{eq:F_lowerbound_2} in Case (4), we have the lower bound
\begin{align*}
\log |\mc F|&\geq   S\log \frac{2^m}{S} - (1-\alpha)S(\log e2^m-\log ((1-\alpha)S-1))-S\\
&\stackrel{(a)}{=}2^{\gamma m}\left(c\alpha(1-\gamma) m+c\log\left(1-\frac{1}{(1-\alpha)c}2^{-\gamma m} \right)-cB\right)\\
&= m2^{\gamma m}\underbrace{\left(c\alpha(1-\gamma) +cm^{-1}\log\left(1-\frac{1}{(1-\alpha)c}2^{-\gamma m} \right)-cB/m\right)}_{A(m)}
\end{align*}
where in (a) we substitute $S=c2^{\gamma m}$ and define $B:=1+(1-\alpha)\log (e/(1-\alpha))+\alpha\log c$. Combining the above inequality with \eqref{eq:converse_F}, we have
\begin{align*}
mA(m)2^{\gamma m}\leq 2^{n(C+\delta)}
\end{align*}
where  $A(m)\rightarrow c\alpha(1-\gamma)$ as $m\rightarrow \infty$. This implies
\begin{align*}
m2^{\gamma m}\leq 2^{n(C+\delta')}
\end{align*}
for $\delta'>\delta$ and sufficiently large $n$ and $m$. On the other hand we have $m2^{\gamma m}\geq 2^{\gamma m}$. This shows that it holds that
\begin{align*}
m\leq \frac{(C+\delta')}{\gamma}n
\end{align*}
for sufficiently large $n$ and $m$, and for all $\delta'>\delta>0$.

\end{proof}

% \section{Concluding remarks}

% Future directions

%\section{Appendix}

\appendix

\subsection{Proof of Proposition \ref{prop:maximal_code}}

\begin{proof}[Proof of Proposition \ref{prop:maximal_code}]
The proof follows that of \cite[Proposition 1]{ahlswede_identification_1989}. For any given $A_1$, we count the number of $A$ such that $|A_1\cap A|\geq \lambda M'$:
\begin{align*}
\sum_{i=\lceil \lambda M'\rceil}^{M'} \binom{|Z|-M'}{M'-i}\binom{M'}{i}
\end{align*}
For $\lambda<1/2$ and $1/\epsilon>6$, the first summand is the maximal one. Then we can upper the above expression as
\begin{align*}
M'\binom{|Z|-M'}{M'-\lceil \lambda M'\rceil}\binom{M'}{\lceil \lambda M'\rceil}\leq M' \binom{|Z|}{M'-\lceil \lambda M'\rceil} {M' \choose \lceil\lambda M'\rceil}=:T.
\end{align*}

As there are $\binom{|Z|}{M'}$ sets of cardinality of $M'$, so if $T<\binom{|Z|}{M'}$, then there must exists at least one $A_2$ such at $|A_1\cap A_2|<\lambda M'$.  Similarly, if $2T< \binom{|Z|}{M'}$, we can find another set $A_3$ with $|A_3|=M'$, $|A_3\cap A_2|<\lambda M'$ and $|A_3\cap A_1|<\lambda M'$. In general, if  it holds that $(N-1)T<\binom{|Z|}{M'}$, we can find $N$ sets with the desired property.  Hence it is easy to see that a family of sets $A_1,\ldots, A_N$ exists with the choice
\begin{align*}
N:=\left\lceil\binom{|Z|}{M'}T^{-1}\right\rceil
\end{align*}
Clearly we have
\begin{align*}
N\geq \binom{|Z|}{M'}T^{-1}
\end{align*}
As shown in \cite[Proposition 1]{ahlswede_identification_1989}, we have
\begin{align*}
\binom{|Z|}{M'}T^{-1}&=H(\lambda, M') M'^{-1}\prod_{i=1}^{\lceil \lambda M'\rceil}\frac{|Z|-M'+i}{M'-\lceil\lambda M'\rceil +i}
\end{align*}
Recall that $M':=\lceil \epsilon |Z|\rceil$. Notice that for $i= 1,\ldots, \lceil \lambda M'\rceil$, we have
\begin{align*}
\frac{|Z|-M'+i}{M'-\lceil\lambda M'\rceil +i}&\geq \frac{|Z|-M'+1}{M'-\lceil\lambda M'\rceil +\lceil \lambda M'\rceil}\\
&\geq \frac{|Z|-M'+1}{M'}\\
&\geq \frac{|Z| -\epsilon |Z|}{M'}\\
&\geq \frac{|Z|-\epsilon|Z|}{2\epsilon|Z|}\\
&=\frac{1-\epsilon}{2\epsilon}
\end{align*}
where the last inequality holds because $M'\leq \epsilon|Z|+1\leq 2\epsilon|Z|$, as it is assumed that $\epsilon |Z|\geq 1$. This implies
\begin{align*}
N\geq H(\lambda, M')M'^{-1}\left( \frac{1-\epsilon}{2\epsilon}\right)^{\lceil \lambda M'\rceil}
\end{align*}
%On the other hand, we have  for all $i= 1,\ldots, \lceil \lambda M'\rceil$
%\begin{align*}
%\frac{|Z|-M'+i}{M'-\lceil\lambda M'\rceil +i}&\leq \frac{|Z|-M'+1}{M'-\lceil\lambda M'\rceil +1}\\
%&\leq\frac{|Z|}{M'-\lceil\frac{1}{2} M'\rceil+1}\\
%&\leq \frac{|Z|}{M'-\frac{1}{2} M'}\\
%&\leq \frac{|Z|}{\epsilon|Z|/2}\\
%&=\frac{2}{\epsilon}
%\end{align*}
%This implies that
%\begin{align*}
%N\leq H(\lambda, M') M'^{-1}\left(\frac{2}{\epsilon}\right)^{\lceil \lambda M'\rceil}+1
%\end{align*}

%
%In the special case when $M'=1$,  we have
%\begin{align*}
%H(\lambda, M')=h_2(\lceil \lambda\rceil) = h_2(1)=0
%\end{align*}
%and
%\begin{align*}
%\binom{|Z|}{M'}T^{-1}&=2^{-M'H(\lambda,M')}\cdot M'^{-1}\prod_{i=1}^{\lceil \lambda M'\rceil}\frac{|Z|-M'+i}{M'-\lceil\lambda M'\rceil +i}=2^{0}\cdot 1 \cdot \frac{|Z|-1+1}{1-1+1}=|Z|
%\end{align*}

\end{proof}

\subsection{Proof of Lemma \ref{lemma:error_probability}}
\label{appendix:error_probability}

We first analyze the false negative error which can happen when $f_j(i)=1$, or equivalently 
\begin{align}
i\in f_j^{-1}[1]
\label{eq:error_analysis_1}
\end{align}
In this case, we have
\begin{align*}
Q_iW(\mc D_j^c)&=\sum_{\ve x\in\mc X^n}Q_i(\ve x)W(\mc D_j^c|\ve x)\\
&\stackrel{\eqref{eq:encoder}}{=}\sum_{\ve x\in\mc X^n}\ve 1_{\ve x\in A_i}\frac{1}{|A_i|} W(\mc D_j^c|\ve x)\\
&=\frac{1}{|A_i|}\sum_{\ve u_k\in A_i}W(\mc D_j^c|\ve u_k)\\
&\stackrel{\eqref{eq:Dj_def}, \eqref{eq:error_analysis_1}}{\leq} \frac{1}{|A_i|}\sum_{\ve u_k\in A_i}W(\mc E_k^c|\ve u_k)\\
&\stackrel{\eqref{eq:good_channel_code}}{\leq}  \frac{1}{|A_i|}\sum_{\ve u_k\in A_i}\delta= \delta
\end{align*}
%Notice that we have  $W(\mc E_k^c|\ve u_k)\leq \delta$ due to \eqref{eq:good_channel_code} and $\mc E_k\subseteq \mc D_j$ for $k$ satisfying $\ve u_k\in A_i$ (because $i\in f_j^{-1}[1]$ in this case) due to the construction in \eqref{eq:Dj_def}. Hence it holds that 
%\begin{align*}
%W(D_j^c|\ve u_k)\leq W(\mc E_k^c|\ve u_k)\leq \delta
%\end{align*}
%for all $\ve u_k\in A_i$. Therefore we conclude
%\begin{align*}
%Q_iW(\mc D_j^c)\leq  \frac{1}{|A_i|}\sum_{\ve u_k\in A_i}\delta= \delta
%\end{align*}
In other words, the false negative error is  always upper bounded by $\delta$, irrespective of $S$.

Now consider the false positive error which can happen when $f_j(i)=0$.  In this case we have
\begin{align*}
Q_iW(\mc D_j)&=\sum_{\ve x\in\mc X^n}Q_i(\ve x)W(\mc D_j|\ve x)\\
&\stackrel{\eqref{eq:encoder}}{=}\sum_{\ve x\in\mc X^n}\ve 1_{\ve x\in A_i}\frac{1}{|A_i|} W(\mc D_j|\ve x)\\
&=\frac{1}{|A_i|}\sum_{\ve u_k\in A_i}W(\mc D_j|\ve u_k)\\
&=\frac{1}{|A_i|}\left( \sum_{\ve u_k\in A_i\bigcap B_j}W(\mc D_j|\ve u_k) +\sum_{\ve u_k\in A_i\backslash B_j}W(\mc D_j|\ve u_k)\right)\\
&\leq \frac{1}{|A_i|}\left( |A_i\cap B_j| +\sum_{\ve u_k\in A_i\backslash B_j}W(\mc D_j|\ve u_k)\right)
\end{align*}
where the last inequality holds as $W(\mc D_j|\ve u_k)\leq 1$. For $\ve u_k\notin B_j$, it holds that $\mc E_k\cap \mc D_j=\emptyset$ by the construction of $\mc D_j$ in \eqref{eq:Dj_def}. This means $\mc D_j\subseteq \mc E_k^c$ hence  for $\ve u_k\in A_i\backslash B_j$ we have
\begin{align*}
W(\mc D_j|\ve u_k)\leq W(\mc E_k^c|\ve u_k)\leq \delta.
\end{align*}
This implies
\begin{align}
Q_iW(\mc D_j)&\leq \frac{|A_i\cap B_j|}{|A_i|} + \frac{\delta |A_i\backslash B_j|}{|A_i|}\nonumber	\\
&\stackrel{(a)}{\leq} \frac{\sum_{\ell \in f_j^{-1}[1]} |A_i\cap A_\ell|}{|A_i|}+\delta\nonumber\\
&\stackrel{(b)}{<}   \frac{\sum_{\ell \in f_j^{-1}[1]} \lambda M'}{M'}+\delta\nonumber\\
&\stackrel{(c)}{\leq} S \lambda +\delta\label{eq:lambda_2}
\end{align}
where $(a)$ follows because of the definition of $B_j$ in \eqref{eq:Dj_def} and an application of the union bound,  $(b)$ follows due to the property of the sets $\{A_i\}_i$ by Proposition \ref{prop:maximal_code}, that is, $|A_i \cap A_j|<\lambda M'$ for all $ i\neq j$. and $(c)$ follows that the Hamming weight of the function $f_j$ is smaller or equal to $S$, namely $|f_j^{-1}[1]|\leq S$.

\subsection{A proof of \textbf{Case 5} in Theorem \ref{thm:achievability}}

As discussed in the proof of Theorem \ref{thm:achievability}, for Case 5 ($S\leq c2^{\gamma m}$), a ``separation approach" (transmission + computation) with an existing error correction code achieves the computation rate of $C$ with the scaling function $L(R,n)=Rn$, hence another proof is not needed. For completeness, here we present a proof following the proof strategy of other cases. As one can see in the proof, the choice of parameters ($\epsilon$ and $\lambda$) make the ID code essentially the same as the original transmission code.

\begin{proof}[A proof of \textbf{Case 5} in Theorem \ref{thm:achievability}] In this case we choose $\epsilon = M^{-1}=2^{(C-\xi)n}$, $\lambda=2^{-Cn}$ and $m=\lceil (C-\xi')n\rceil$ for some $\xi'>\xi$.  This choice of $\epsilon$ satisfies \eqref{eq:epsilon_const} for large enough $n$. To proceed in the same way as for other cases, we first lower bound $\tilde N$ in \eqref{eq:m} as
\begin{align*}
\tilde N&=H(\lambda,M')M'^{-1}\left( \frac{1-\epsilon}{2\epsilon}\right)^{\lceil\lambda M'\rceil} \nonumber\\
&\geq 2^{-M'}M'^{-1}\left( \frac{1-\epsilon}{2\epsilon}\right)^{\lceil\lambda M'\rceil} 
\end{align*}
The choice of $\epsilon$ and $\lambda$ implies that $M'=1$ and $\lceil\lambda M'\rceil=1$. We can further lower bound
\begin{align*}
\tilde N\geq \frac{1}{2}\cdot \frac{1-\epsilon}{2\epsilon}
\end{align*}
Now we prove that the choice $m=\lceil (C-\xi')n\rceil$ satisfies \eqref{eq:m} by proving 
\begin{align*}
2^m\leq \frac{1}{2}\cdot \frac{1-\epsilon}{2\epsilon}
\end{align*}
Substitute $m$ in the above inequality and take logarithm on both sides, we obtain
\begin{align*}
\lceil (C-\xi')n\rceil\leq \log 1/\epsilon +\log (1-\epsilon) - \log 4
\end{align*}
which can be satisfied if  (substituting $\epsilon=2^{-(C-\xi)n}$)
\begin{align*}
(C-\xi')n+ 1\leq  (C-\xi)n +\log (1-2^{-(C-\xi)n}) - \log 4.
\end{align*}
This is equivalent to
\begin{align*}
\xi-\xi'\leq \frac{\log (1-2^{-(C-\xi)n}) - \log 4-1}{n}
\end{align*}
Notice that the RHS approaches $0$ and $n\rightarrow \infty$ where the LHS is negative as long as $\xi'>\xi$, which verifies \eqref{eq:m}.

To upper bound the error probability, we have
\begin{align*}
\lambda_2&= S\lambda+\delta\\
&= c2^{\gamma m}2^{-Cn}+\delta\\
&=c2^{\gamma \lceil (C-\xi')n\rceil-Cn}+\delta\\
&\leq c2^{\gamma (C-\xi')n+\alpha-Cn}+\delta\\
&=c2^{-(1-\gamma)Cn-\xi'n+\alpha}+\delta
\end{align*}
Since $\gamma\in(0,1]$, the error probability can be made arbitrarily small for large enough $n$ and a small enough $\delta$.

\end{proof}

\subsection{Lemma \ref{lemma:gilbert_cw} and its proof}

The following lemma is used to prove Lemma \ref{lemma:bcf_to_id}, which will be used in the proof of the converse result in Theorem \ref{thm:converse_subset}.

\begin{lemma}[Gilbert's bound for constant weight sequences]\label{lemma:gilbert_cw}
Let  $T, S$ be two natural numbers such that $1\leq S\leq T/2$. For any $\alpha\in(0,1]$,  there exists a set $G$ of $T$-length binary sequences with the following properties
\begin{itemize}
%\item for any $a\in G$, $a$ has Hamming weight $S$, and $d_H(a,b)\geq d$ for any $a,b\in G, a\neq b$
\item  for any $a,b\in G$, $|a\cap b|\leq \alpha S$
\item $\log |G|\geq \log {T\choose S}-\log {T\choose \lceil (1-\alpha)S\rceil-1 }-S$
\end{itemize}
\end{lemma}

\begin{proof}
By the Gilbert bound for constant weight sequences with a weight $S$, we know from  \cite[Theorem 7]{graham_lower_1980} that for any integer $0\leq d\leq S$, there exists a set $G$ of binary sequences where $d_H(a,b)\geq 2d$ for any $a,b\in G, a\neq b$ and
\begin{align*}
|G|\geq \frac{{T\choose S}}{\sum_{i=1}^{d-1}{S \choose i}{T-S \choose i}}
\end{align*}

To see the first claim, notice that two sequences $a,b$ with weight $S$ and Hamming distance $2\delta$ has an intersection $|a\cap b| = S-\delta$. To see this, we count the total weights of the two sequences in two different ways. First note that the total weights of the two sequences is $2S$, as both sequences have weight $S$. On the other hand, the total weights can also be expressed as $2|a\cap b|+2\delta$, where $2|a\cap b|$ is the number of positions where the two sequences take the value $1$, and $2\delta$ is the number of positions where only one of the two sequences takes the value $1$.  We let $d=\lceil (1-\alpha)S\rceil$.  Then it holds that $|a\cap b|=S-d_H(a,b)/2\leq  S-d\leq \alpha S$ as $d\geq (1-\alpha)S$.

To show the lower bound, notice that we have for $i=1,\ldots, d-1$
\begin{align*}
\sum_{i=1}^{d-1}{S \choose i}{T-S \choose i}&\leq \sum_{i=1}^{d-1}{S \choose i}{T \choose i}\\
&\stackrel{(a)}{\leq}  \sum_{i=1}^{d-1}{S \choose i}{T \choose d-1}\\
&\leq 2^S {T\choose d-1}
\end{align*}
where (a) holds because for $i=1,\ldots, d-1$ with $d-1\leq (1-\alpha)S\leq S\leq T/2$, we have ${T\choose i}\leq {T\choose d-1}$. 

\end{proof}

\bibliographystyle{IEEEtran}
\bibliography{BFC}

%\end{thebibliography}

\end{document}